\documentclass[11pt,a4paper]{article}
\usepackage[margin=1in]{geometry}
\usepackage{amsmath,amssymb,amsthm}
\usepackage{graphicx}
\usepackage{float}   
\usepackage[colorlinks=true,linkcolor=blue,citecolor=blue]{hyperref}
\usepackage{booktabs}
\usepackage{caption}
\newtheorem{proposition}{Proposition}
\newtheorem{lemma}{Lemma}
\newtheorem{corollary}{Corollary}
\theoremstyle{definition}
\newtheorem{assumption}{Assumption}
\theoremstyle{remark}
\newtheorem{remark}{Remark}

\title{Sharp Bounds on the Mean Efficiency of a Fluctuating Machine}

\author{%
Badr Farih\\[2pt]
\normalsize \texttt{badr.farih1@usmba.ac.ma}%
}

\date{}

\begin{document}
\maketitle
\vspace{-2.2em}

\begingroup
\leftskip=0.06\textwidth \rightskip=0.06\textwidth
\noindent\textbf{Abstract.}
The efficiency of a machine operating at the scale of thermal fluctuations is a
random variable whose statistics are ordinarily obtained from a dynamical model,
and which conventionally has no moments of any order: the input heat in the
denominator of $-W/Q_{\rm h}$ fluctuates through zero. We work instead with the
exergetic ratio $\eta=W/(W+T_0S)$, whose denominator vanishes only with its
numerator, so that for non-negative dissipation it lies in $[0,1]$ pointwise and
every moment exists --- and we ask what the energy budget alone then determines
about its mean. With
$\alpha=T_0\langle S\rangle/W$ the mean dissipation per unit useful work and
$\sigma^{2}$ the relative variance of the dissipation,
\[
\frac{1}{1+\alpha}\;<\;\langle\eta\rangle\;\le\;
\frac{\sigma^{2}}{1+\sigma^{2}}+\frac{1}{(1+\sigma^{2})[1+\alpha(1+\sigma^{2})]},
\]
with both ends sharp and no distributional assumption. The lower bound is
Jensen's inequality --- fluctuating dissipation always raises the mean efficiency
above its deterministic value --- and the mean alone determines nothing above it.
The upper bound is attained by an \emph{intermittently reversible} law, which
dissipates nothing at all in a fraction $\sigma^{2}/(1+\sigma^{2})$ of
realisations; it is a moment-problem extremal rather than a realised machine, but
a stability result turns it into a prediction --- a device measured near the bound
must operate intermittently, testable against the trajectory record alone. A
third moment lifts the floor, closing the bracket entirely at the smallest third
moment the first two permit. Fixed delivered work is not required: when it too
fluctuates the bounds hold with the moments taken on the ratio $T_0S/W$, whose
mean is not the ratio of the separate means. In that setting the
thermodynamic uncertainty relation, applied to the work current, converts the
ceiling into a precision--efficiency frontier: with $\epsilon_w^{2}$ the relative
variance of the delivered work and $\gamma=2k_BT_0/\langle W\rangle$ scaled by the
fraction $r\le1$ of the dissipation the observed current accounts for, the mean
efficiency is at most
$\varsigma^{2}/(1+\varsigma^{2})+\epsilon_w^{2}/\{(1+\varsigma^{2})
[\epsilon_w^{2}+r\gamma(1+\varsigma^{2})]\}$, so a machine with more reproducible
output has a strictly lower efficiency ceiling. At $\varsigma^{2}=0$ this is
exactly the known bound on a molecular motor's ratio-of-means efficiency, which
is thereby identified as the zero-dissipation-variance member of a family and
shown to be unsafe when applied to the mean of the fluctuating ratio. Admitting
negative dissipation restores a pole inside the physical region and the mean
diverges again, so the non-negativity is not a restriction narrowing the theory
but the condition under which a moment-based account of stochastic efficiency
exists at all. Inside the interval lies the maximum-entropy benchmark
$\alpha^{-1}e^{1/\alpha}E_{1}(1/\alpha)$. We close by working the bounds out for
a molecular motor with futile cycles, observed until it has delivered a fixed
number of steps: there the dissipation cannot fall below the reversible cost of
the delivered work, and that floor $b$ sharpens the ceiling to
$q/(1+\alpha b)+(1-q)/(1+\alpha c)$ with $q=\sigma^{2}/[\sigma^{2}+(1-b)^{2}]$
and $c=1+\sigma^{2}/(1-b)$, removing between $40$ and $67$ per cent of the
interval width. The sharpened ceiling is saturated when slips are rare, so the
extremal law is not only an idealisation: it is the operating regime of a motor
whose losses are rare and complete.

\vspace{0.6em}
\noindent\textbf{keywords:} stochastic thermodynamics;
efficiency fluctuations; maximum-entropy inference; entropy production;
thermodynamic uncertainty relation; molecular machines.
\par\endgroup
\section*{1. Introduction}

At macroscopic scales the efficiency of a thermodynamic process is a number. At
the scale of a single molecular machine it is not: the work delivered and the
heat dissipated in any one realisation are random variables, and so is their
ratio. The point is not academic. In a Brownian Carnot engine built from an
optically trapped colloidal particle, the measured efficiency fluctuates from
cycle to cycle strongly enough that the Carnot bound is exceeded over a small
number of non-equilibrium cycles~\cite{Martinez2016}. Any statement about
``the'' efficiency of such a device is therefore a statement about a
distribution, and the interesting question is which features of that
distribution can be pinned down, and from what information.

Stochastic thermodynamics answers this question comprehensively when the
dynamics are known~\cite{Seifert2012}. Verley \emph{et al.} compute the
large-deviation function of the fluctuating efficiency and find the striking
result that, in the long-time limit, the Carnot efficiency is the \emph{least}
likely outcome rather than an unattainable ceiling~\cite{Verley2014}.
Muratore-Ginanneschi and Schwieger obtain the efficiency at maximum power for
sub-micron engines by solving an optimal-transport problem for the driving
protocol~\cite{Muratore2015}. These results are sharp, and they share a
prerequisite: a specified model of the machine --- its potential, its rates, its
protocol. Given the dynamics, the distribution of efficiency follows.

Often the dynamics are exactly what is missing. A coarse-grained biological
machine, a device characterised only through its aggregate energy budget, or a
process observed at a level of description far above its microscopic degrees of
freedom may offer no tractable model, while still offering an energy budget: the
useful work delivered, the temperature of the environment, and the first one or
two moments of the dissipation. This paper asks what such data determine about
the mean efficiency, and answers it completely.

We treat the entropy production $S$ as an uncertain quantity and assign it the
distribution that is maximally noncommittal with respect to everything not
specified~\cite{Jaynes1957}. Constraining only $S\ge0$ and its mean
$\langle S\rangle$, the maximum-entropy prior is exponential. Combining it with
the efficiency implied by energy accounting, $\eta = W/(W+T_0S)$, yields a mean
efficiency that depends on a single dimensionless group,
\[
\alpha \;=\; \frac{T_0\langle S\rangle}{W},
\]
the ratio of mean dissipated energy to useful work, and admits a closed form in
terms of the exponential integral,
$\langle\eta\rangle = \alpha^{-1}e^{1/\alpha}E_1(1/\alpha)$. The prior's scale is absorbed entirely into $\alpha$, so no
arbitrary choice of entropy unit survives into the result.

The central result is a pair of bounds. Because $\eta$ is convex in $S$,
Jensen's inequality gives $\langle\eta\rangle>(1+\alpha)^{-1}$: at fixed mean
dissipation, \emph{fluctuations raise the mean efficiency} above the value a
machine would have if it dissipated its mean amount every time. In the opposite
direction, adding the relative variance
$\sigma^{2}=\mathrm{Var}(S)/\langle S\rangle^{2}$ caps it, and the two together
determine $\langle\eta\rangle$ to within
\[
\frac{1}{1+\alpha}\;<\;\langle\eta\rangle\;\le\;
\frac{\sigma^{2}}{1+\sigma^{2}}+\frac{1}{(1+\sigma^{2})[1+\alpha(1+\sigma^{2})]} ,
\]
sharp at both ends, from two measured numbers and no distributional assumption.
The interval closes as $\sigma^{2}\to0$, with width
$\sigma^{2}\alpha^{2}/(1+\alpha)^{2}$.

The bounds take a further constraint from the dynamics once the conditioning on
the delivered work is removed. The thermodynamic uncertainty
relation~\cite{Barato2015,Gingrich2016,Horowitz2017}, applied to the work current
rather than to the dissipation, bounds the dissipation ratio below by the
precision of the output; the moment ceiling is strictly decreasing in that ratio
and the interval's width strictly increasing, so a single inequality both caps the
mean efficiency and forbids the interval from closing. The ceiling reaches the
machine through two measured relative variances and the ratio
$\gamma=2k_BT_0/\langle W\rangle$ alone, and at zero dissipation variance it
reduces exactly to the known universal bound on a molecular motor's
efficiency~\cite{Pietzonka2016} --- which it thereby places inside a family, and
corrects, since that bound constrains the ratio of means and is not valid for the
mean of the ratio.

The upper bound is attained by a two-point law with an atom at $S=0$: the most
efficient machine compatible with a given mean and variance is
\emph{intermittently reversible}, dissipating nothing at all in a fraction
$\sigma^{2}/(1+\sigma^{2})$ of its realisations. This converts the bound into a
prediction about behaviour rather than a mere inequality, and one that can be
checked against a trajectory record without measuring efficiency at all.

Section~9 puts the bounds to work on a machine. For a molecular motor with
futile cycles, observed until it has delivered a fixed number of steps, the
delivered work is deterministic and the dissipation is negative-binomial, so
$\alpha$, $\sigma^{2}$ and $\langle\eta\rangle$ are all exact. The interval is
informative there, and one further piece of thermodynamics sharpens it: the
motor cannot dissipate less than the reversible cost of the work it delivers, and
that floor moves the extremal atom away from zero dissipation, removing between
$40$ and $67$ per cent of the width. The sharpened ceiling is then saturated in
the rare-slip limit, in which the dissipation becomes two-valued --- so the
extremal law of the moment problem is realised by a machine defined by rates,
rather than remaining an idealisation.

Two further findings frame the bounds. First, the variance is genuinely needed:
embedding the exponential prior in the unit-mean Gamma family shows that at fixed
mean the mean efficiency sweeps the whole interval between the Jensen bound and
unity, so the mean dissipation alone constrains it \emph{only from below}. The
upper bound in the display above is exactly what the second moment buys. Second,
under the maximum-entropy prior the mean efficiency is available in closed form
and interpolates between $1-\alpha$ near reversibility and
$(\ln\alpha-\gamma)/\alpha$ in the strongly dissipative limit --- decaying
logarithmically more slowly than the deterministic estimate. That value is an
interior point of the admissible interval, as a noncommittal choice should be.

One point of framing should be made before the clarifications, because it decides
how the central assumption is to be read. The conventional stochastic efficiency
$-W/Q_{\rm h}$ has no moments of any order: the input heat in the denominator
fluctuates through zero and the distribution acquires inverse-square
tails~\cite{Holubec2022,Polettini2015}. That is why the field characterises
stochastic efficiency by its distribution and its large-deviation function rather
than by an average. The exergetic ratio used here has a denominator that vanishes
only together with its numerator, so with non-negative dissipation it lies in
$[0,1]$ pointwise and every moment exists; Proposition~\ref{prop:exist}
identifies the single route back to the pathology, and
Assumption~\ref{as:pos} closes it. The assumption is accordingly not a
restriction narrowing a general theory. It is, with the choice of figure of
merit, what makes a moment-based account of stochastic efficiency available at
all --- and the bounds of Sections~3--6 are what that buys.

Three clarifications of scope are needed at the outset. The bounds are obtained
from the classical Chebyshev--Markov moment machinery~\cite{KarlinStudden,
Vandenberghe2007}, applied here to stochastic efficiency; we claim novelty for
the application and for the physical reading of the extremal law, not for the
method, whose extremal two-point structure is standard. The bounds are
distribution-free and hold for any machine satisfying the stated energy
accounting; the closed-form value inside them, by contrast, is a reference under
stated ignorance rather than a prediction for any particular machine, and a
measured deviation from it quantifies how much structure the real system
possesses beyond its mean dissipation. And although the weight $e^{-S/\langle S\rangle}$ coincides in form with
the kernel of the fluctuation theorems when $\langle S\rangle = k_B$, it enters
here as an epistemic prior rather than as a ratio of path measures; we claim no
fluctuation-theorem result, and the formal resemblance should not be read as one.

The paper is organised as follows. Section~2 fixes the efficiency definition and
the information state. Section~3 establishes the prior-free lower bound and shows
that the mean alone determines nothing above it. Section~4 derives the upper
bound from the variance, identifies the extremal machine, assembles the
admissible interval, and shows that its width carries a thermodynamic floor set
by the uncertainty relations --- a floor that is attained, not merely respected,
when the dissipation is Gaussian. Section~5 shows that a third moment lifts the lower bound
and closes the bracket entirely at the smallest third moment the first two
permit. Section~6 removes the conditioning on the delivered work, shows that the
bounds survive intact, and --- since the work is then free to fluctuate --- uses
the uncertainty relation on the work current to convert the ceiling into a
precision--efficiency frontier. Section~7 evaluates the maximum-entropy benchmark that
sits inside the interval, in closed form. Section~8 relates the construction to existing results on
efficiency fluctuations, Section~9 works the bounds out for a molecular
motor with futile cycles, and Section~10 discusses the regime of validity.

\section*{2. Setup}

\subsection*{2.1 Setting and assumptions}

We consider a process that performs a specified task --- delivering useful work
$W>0$ --- in contact with a thermal environment at temperature $T_0$, and that
produces entropy $S$ in doing so. Three assumptions fix the setting.

\begin{assumption}[Conditioning on the task]
\label{as:task}
All expectations below are conditional on the delivered work $W$; the prior is a
prior for $S$ given $W$, with conditional mean $\langle S\mid W\rangle = \mu$.
\end{assumption}

\noindent
In a single realisation of a small machine both the work output and the
dissipation fluctuate, and in general they are correlated. We do not assume
otherwise. Assumption~\ref{as:task} states instead what object is being computed:
a conditional expectation $\langle\eta\mid W\rangle$, the mean efficiency of the
machine \emph{given} that it delivered the required work. This is the relevant
quantity for a device with a specified function, and it requires no claim about
the joint law of $(W,S)$. Recovering the unconditional mean would require that
joint law and is not attempted here. For readability we write
$\langle\,\cdot\,\rangle$ for $\langle\,\cdot\mid W\rangle$ throughout.

\begin{assumption}[Non-negative entropy production]
\label{as:pos}
$S\ge 0$.
\end{assumption}

\noindent
This assumption is constitutive rather than cosmetic, and we make its role
explicit. The fluctuation theorems establish that individual trajectories may
have $S<0$, with probability exponentially small in $|S|/k_B$~\cite{Seifert2012};
that is precisely what makes small-engine efficiency interesting. It is therefore
worth asking what happens if the restriction is dropped. The answer is that the
quantity under study ceases to exist.

\begin{proposition}[Existence of a mean efficiency]
\label{prop:exist}
Let $S$ have a prior with a density that is positive and continuous in a
neighbourhood of $S_\ast = -W/T_0$. Then $\langle\eta\rangle = +\infty$.
\end{proposition}

\begin{proof}
$S_\ast$ is the zero of $E_{\mathrm{cons}} = W+T_0S$, so \eqref{eq:eta} has a
simple pole there. Writing $p(S_\ast)=p_\ast>0$, the integrand behaves as
$p_\ast W/\!\left[T_0(S-S_\ast)\right]$ as $S\to S_\ast$, whose integral over any
one-sided neighbourhood diverges logarithmically. Excluding the region
$E_{\mathrm{cons}}\le 0$ does not help: the divergence is approached from within
the physical region.
\end{proof}

\begin{remark}
Proposition~\ref{prop:exist} explains a feature of the literature rather than
contradicting it. Studies of stochastic efficiency work with the distribution of
$\eta$ and its large-deviation function rather than with its
mean~\cite{Verley2014}, and the reason is visible here: for a fluctuating
denominator the mean is generically ill-defined. For the conventional efficiency
$-W/Q_{\rm h}$ this is not a technicality but a known obstruction: the input heat
in the denominator fluctuates through zero, the distribution acquires
$\rho(\eta)\sim\eta^{-2}$ tails, and every moment
diverges~\cite{Holubec2022,Polettini2015}. The exergetic ratio used here does not
share the defect. With $W$ and $S$ both non-negative it lies in $[0,1]$
pointwise, so all its moments exist (Proposition~\ref{prop:jointexist}), and
Proposition~\ref{prop:exist} identifies the sole route by which the pathology
returns --- a dissipation admitting negative values. Assumption~\ref{as:pos}
closes that route. The definition and the assumption together are therefore not a
restriction imposed on a general theory; they are what makes a moment-based
description of stochastic efficiency available at all. In practice the assumption
holds when $S$ is coarse-grained --- over a cycle, an ensemble, or an observation
window --- so that negative excursions are absent.
\end{remark}

\subsubsection*{When is Assumption~\ref{as:pos} justified?}

Since the assumption is load-bearing, it deserves a criterion rather than an
appeal to coarse-graining. Two distinct situations make it legitimate, and they
are worth separating because only one of them constrains $\sigma^{2}$.

\emph{(i) Structural non-negativity.} If the dissipation has no reverse channel
--- if $S$ is defined as the integral of a positive-definite quantity, or the
driving is strong enough that reverse transitions are negligible --- then
$S\ge0$ holds \emph{exactly}, whatever the size of its fluctuations. This is the
important case, and it removes an objection the reader may otherwise form: large
relative fluctuations and strict non-negativity are not in conflict. An
exponentially distributed dissipation has $\sigma^{2}=1$, a
$\mathrm{Gamma}(1/2)$ dissipation has $\sigma^{2}=2$, and both are supported on
$[0,\infty)$. Assumption~\ref{as:pos} is a statement about the structure of the
process, not about how strongly it fluctuates.

\emph{(ii) Aggregation.} If instead $S$ is a sum of many weakly correlated
contributions --- the usual reading of ``coarse-grained over a cycle or an
observation window'' --- then $S$ is approximately Gaussian and
\begin{equation}
\Pr\left[S<0\right]\;\approx\;\Phi\!\left(-1/\sigma\right),
\label{eq:negprob}
\end{equation}
$\Phi$ being the standard normal distribution function and $\sigma^{2}$ the same
relative variance that appears throughout this paper. This is negligible only
for small $\sigma$: \eqref{eq:negprob} gives $8\times10^{-4}$ at
$\sigma^{2}=0.1$, $2\times10^{-2}$ at $\sigma^{2}=0.25$, and $0.16$ at
$\sigma^{2}=1$. The aggregation route therefore justifies
Assumption~\ref{as:pos} only in the near-deterministic regime, and we do not rely
on it elsewhere.

\emph{Where the assumption fails.} Trajectory-level total entropy production of
a system near equilibrium, with $\sigma^{2}$ of order unity, violates
Assumption~\ref{as:pos} outright; the present framework does not apply there,
and the efficiency of such a system must be studied through its distribution
rather than its mean.

Finally, it is worth recording that \emph{existence} requires strictly less than
the assumption states. Proposition~\ref{prop:exist} fails only if the density
reaches the pole at $S_\ast=-W/T_0$, so $\langle\eta\rangle$ exists whenever
$S>-W/T_0$ almost surely. The integral fluctuation theorem bounds excursions in
that direction --- $\langle e^{-S/k_B}\rangle=1$ with Markov's inequality gives
$\Pr[S\le-s]\le e^{-s/k_B}$, hence $\Pr[S\le-W/T_0]\le e^{-W/(k_BT_0)}$ --- but
this reassures rather than substitutes: the divergence of
Proposition~\ref{prop:exist} is logarithmic and approached from within the
physical region, so a small probability at the pole is not enough. What
Assumption~\ref{as:pos} buys beyond existence is the bounds themselves, whose
majorant and minorant arguments use $x\ge0$ throughout.

\noindent
One might still object that restricting the prior to $S\ge 0$ is what produces
the result. It is not. Appendix~A constructs the prior that admits
$S<0$ while respecting the integral fluctuation theorem
$\langle e^{-S/k_B}\rangle=1$, obtained by imposing that identity as a
maximum-entropy constraint alongside the mean; the second law $\mu\ge0$ emerges
from the constraint rather than being assumed. Conditioned on the physical region
$S\ge0$ --- which Proposition~\ref{prop:exist} shows is unavoidable --- that prior
and the exponential prior \eqref{eq:prior} give mean efficiencies agreeing to
within a few per cent at equal mean dissipation, converging as $\mu/k_B$ grows
(Table~\ref{tab:ftprior}). The restriction to $S\ge0$ is required for the mean to
exist; the particular prior on that half-line is not what drives the answer.

\begin{assumption}[Known mean dissipation]
\label{as:mean}
The mean $\langle S\rangle = \mu < \infty$ is known; nothing further about the
distribution of $S$ is known.
\end{assumption}

\noindent
Assumption~\ref{as:mean} is the information state whose consequences this paper
works out. It is deliberately weak: it is what remains when the dynamics of the
machine are unavailable but its energy budget is not.

\subsection*{2.2 Efficiency from energy accounting}

By the Gouy--Stodola theorem the energy irreversibly lost to the environment is
$T_0S$. The energy consumed by the process therefore decomposes as
\begin{equation}
E_{\mathrm{cons}} \;=\; \underbrace{W}_{\text{delivered}} \;+\;
\underbrace{T_0 S}_{\text{dissipated}},
\label{eq:budget}
\end{equation}
and the efficiency --- the fraction of what was spent that emerged as useful
work --- is
\begin{equation}
\eta(S) \;=\; \frac{W}{W + T_0 S}.
\label{eq:eta}
\end{equation}
Equation~\eqref{eq:eta} has the properties one wants of an efficiency without
further stipulation: it is bounded in $(0,1]$ for every $S\ge0$, attains $1$ only
at $S=0$, decreases monotonically in the dissipation, and tends to $0$ rather
than becoming negative as $S\to\infty$. In particular the domain of $S$ needs no
truncation.

\subsection*{2.3 Why not the linear definition}

A common alternative fixes the input energy $E_{\mathrm{in}}$ and writes
$\tilde\eta(S) = 1 - T_0S/E_{\mathrm{in}}$. For the present purpose it is
unsuitable, for a reason worth stating explicitly rather than leaving to the
reader.

Because $\tilde\eta$ is affine in $S$ and expectation is linear,
\begin{equation}
\langle\tilde\eta\rangle \;=\; 1 - \frac{T_0\langle S\rangle}{E_{\mathrm{in}}}
\label{eq:linear}
\end{equation}
for \emph{every} distribution of $S$ with mean $\langle S\rangle$. The
distribution enters only through its first moment; no assignment of a prior can
change the answer, and no question about the influence of dissipation
\emph{fluctuations} on efficiency can be posed within that definition. This is
not a defect of $\tilde\eta$ for its usual purposes, but it makes it inert here.

\begin{remark}[The two definitions on a heat engine]
\label{rem:twobath}
The distinction is not formal. Consider an engine absorbing $Q_h$ at $T_h$,
rejecting heat at $T_c$ and returning to its initial state each cycle, so that
$S=-Q_h/T_h+Q_c/T_c$ and $W=Q_h\eta_{\mathrm{C}}-T_cS$ with
$\eta_{\mathrm{C}}=1-T_c/T_h$. Its first-law efficiency,
\[
\frac{W}{Q_h}\;=\;\eta_{\mathrm{C}}-\frac{T_c}{Q_h}\,S ,
\]
is affine in $S$: the supply $Q_h$ is held fixed and the output absorbs the
fluctuation, so by \eqref{eq:linear} its mean depends on $\langle S\rangle$ alone.
Taking the dead state at $T_0=T_c$, however, the exergy supplied is
$Q_h\eta_{\mathrm{C}}=W+T_cS$, and the second-law efficiency is
\[
\frac{W}{W+T_cS} ,
\]
which is \eqref{eq:eta} exactly. The two describe the same engine; they differ in
which quantity is held fixed. Definition~\eqref{eq:eta} is the one appropriate to
a machine required to deliver a specified output at a fluctuating cost, and it is
the only one of the two for which the distribution of $S$ matters at all.
\end{remark}
A further inconvenience is that $\tilde\eta$ becomes negative for
$S > E_{\mathrm{in}}/T_0$, so any averaging over $S$ requires the domain to be
truncated at that point, and results then depend on a boundary that carries no
physical content. Definition~\eqref{eq:eta} avoids both issues, and its
non-linearity is what gives the distribution of $S$ --- and hence the question
this paper asks --- something to do.

\subsection*{2.4 Reduction to a single parameter}

Rescale the entropy production by its mean, $x = S/\mu$, so that the prior on $x$
has unit mean, and define
\begin{equation}
\boxed{\;\alpha \;\equiv\; \frac{T_0\,\langle S\rangle}{W}\;}
\label{eq:alpha}
\end{equation}
the ratio of mean dissipated energy to useful work delivered. Substituting into
\eqref{eq:eta},
\begin{equation}
\eta(x) \;=\; \frac{1}{1+\alpha x},
\qquad x\ge 0 .
\label{eq:etax}
\end{equation}
The problem now contains one dimensionless parameter. Small $\alpha$ is the
near-reversible regime, in which dissipation is a small correction to the work
delivered; $\alpha \gtrsim 1$ is the regime in which a machine wastes at least as
much energy as it delivers. Note that $\alpha$ measures a loss fraction and not a
system size: as discussed in Section~10.1, what makes fluctuation corrections
appreciable is the relative variance of $S$, not the magnitude of $\alpha$.

\subsection*{2.5 The maximum-entropy prior}

Under Assumptions~\ref{as:pos} and~\ref{as:mean} the distribution that is
maximally noncommittal with respect to everything not specified is obtained by
maximising the differential entropy $H[p]=-\int_0^\infty p\ln p\,\mathrm{d}S$
subject to normalisation and the mean constraint~\cite{Jaynes1957}. Introducing
multipliers $\lambda_0,\lambda_1$ and setting the variation of
$H-\lambda_0\!\left(\int p-1\right)-\lambda_1\!\left(\int Sp-\mu\right)$ to zero
gives $-\ln p(S)-1-\lambda_0-\lambda_1 S = 0$, so $p(S)\propto e^{-\lambda_1 S}$.
The constraints fix both multipliers, yielding
\begin{equation}
p(S) \;=\; \frac{1}{\mu}\,e^{-S/\mu},
\qquad\text{or, in the scaled variable,}\qquad
p(x) \;=\; e^{-x}.
\label{eq:prior}
\end{equation}
The exponential form is thus derived, not posited: it is the unique consequence
of constraining support and mean and nothing else. What is \emph{not} derived is
the value of $\mu$, and it is worth being precise about the status of that
parameter.

\subsection*{2.6 Scale invariance}

\begin{lemma}[The prior scale is not a free physical assumption]
\label{lem:scale}
Write $\mu = \lambda k_B$ for an arbitrary $\lambda>0$. Then the mean efficiency
depends on $\lambda$, $T_0$ and $W$ only through the combination
$\alpha = \lambda k_B T_0/W$.
\end{lemma}

\begin{proof}
Under \eqref{eq:prior} the density of $x=S/\mu$ is $e^{-x}$, independent of
$\mu$. All $\mu$-dependence of $\langle\eta\rangle=\int_0^\infty
e^{-x}(1+\alpha x)^{-1}\mathrm{d}x$ therefore resides in $\alpha$, which contains
$\lambda$ only as the stated product.
\end{proof}

\section*{3. What the mean dissipation determines}

This section establishes what the mean dissipation alone settles about the mean
efficiency. The answer is one-sided: a sharp lower bound that holds for every
distribution, and no upper bound whatever. Both halves are needed later --- the
first is the floor of the bracket assembled in Section~4, and the second is why a
second measurement is required at all. Along the way we introduce the
maximum-entropy family used throughout, in which the closed-form benchmark of
Section~7 is the member selected by knowing the mean and nothing else.

\subsection*{3.1 What holds for every prior}

\begin{proposition}[Fluctuations raise the mean efficiency]
\label{prop:jensen}
Let $S\ge0$ have any distribution with $\langle S\rangle=\mu$, and let
$\alpha=T_0\mu/W$. Then
\begin{equation}
\langle\eta\rangle \;\ge\; \frac{1}{1+\alpha} \;=\; \eta_{\mathrm{det}},
\label{eq:jensen}
\end{equation}
with equality if and only if $S=\mu$ almost surely.
\end{proposition}

\begin{proof}
In the scaled variable $x=S/\mu$ the efficiency is $\eta(x)=(1+\alpha x)^{-1}$,
with $\eta''(x)=2\alpha^2(1+\alpha x)^{-3}>0$ for $x\ge0$, so $\eta$ is strictly
convex there. Jensen's inequality gives
$\langle\eta(x)\rangle\ge\eta(\langle x\rangle)=\eta(1)=(1+\alpha)^{-1}$, and
strict convexity makes the inequality strict unless $x$ is degenerate.
\end{proof}

\noindent
Proposition~\ref{prop:jensen} requires no assumption beyond $S\ge0$ and a finite
mean: no exponential prior, no maximum-entropy argument, no distributional form
whatsoever. It is the statement that a machine whose dissipation fluctuates is,
on average, more efficient than one that dissipates its mean deterministically,
and it is the reason for the $\ln\alpha$ enhancement found in
Proposition~\ref{prop:large}. Any claim this paper makes that does not reduce to
\eqref{eq:jensen} is contingent on the prior.

\subsection*{3.2 A principled family of alternatives}

To measure that contingency we need alternative priors that are themselves
principled rather than arbitrary. The natural choice is the family obtained by
adding exactly one constraint to the information state. Maximising $H[p]$ on
$x\ge0$ subject to fixed $\langle x\rangle$ \emph{and} fixed $\langle\ln
x\rangle$ gives $p(x)\propto x^{k-1}e^{-bx}$, and imposing unit mean fixes $b=k$:
\begin{equation}
p_k(x)\;=\;\frac{k^k}{\Gamma(k)}\,x^{k-1}e^{-kx},
\qquad k>0,\qquad \langle x\rangle=1.
\label{eq:gamma}
\end{equation}
The Gamma family is thus not an ad hoc perturbation but the maximum-entropy
family for one additional piece of information, with $k$ in one-to-one
correspondence with the value of $\langle\ln x\rangle$. It contains the prior of
Section~2 at $k=1$; smaller $k$ describes a more variable dissipation and larger
$k$ a more predictable one, at fixed mean.

\begin{proposition}[Mean efficiency across the family]
\label{prop:gamma}
With \eqref{eq:gamma},
\begin{equation}
\langle\eta\rangle_k(\alpha)\;=\;\left(\frac{k}{\alpha}\right)^{\!k}
U\!\left(k,k,\frac{k}{\alpha}\right),
\label{eq:gammares}
\end{equation}
$U$ being the Tricomi confluent hypergeometric function. At $k=1$,
$U(1,1,z)=e^zE_1(z)$ recovers Proposition~\ref{prop:main} of Section~7.
\end{proposition}

\begin{proof}
Substituting $t=\alpha x$,
\[
\langle\eta\rangle_k=\left(\frac{k}{\alpha}\right)^{\!k}\frac{1}{\Gamma(k)}
\int_0^\infty t^{k-1}e^{-(k/\alpha)t}(1+t)^{-1}\mathrm{d}t ,
\]
and comparison with $U(a,b,z)=\Gamma(a)^{-1}\int_0^\infty
e^{-zt}t^{a-1}(1+t)^{b-a-1}\mathrm{d}t$ forces $a=k$, $z=k/\alpha$ and, from
$(1+t)^{b-a-1}=(1+t)^{-1}$, $b=a=k$.
\end{proof}

\subsection*{3.3 The direction and the size of the dependence}

\begin{proposition}[Monotonicity]
\label{prop:mono}
For every $\alpha>0$, $\langle\eta\rangle_k(\alpha)$ is strictly decreasing in
$k$ on $(0,\infty)$.
\end{proposition}

\begin{proof}
Fix $0<k_1<k_2$ and put $m=k_2-k_1>0$. Let $G_1\sim\Gamma(k_1,1)$ and
$G_m\sim\Gamma(m,1)$ be independent and set $G_2=G_1+G_m\sim\Gamma(k_2,1)$, so
that $X_j:=G_j/k_j$ has the unit-mean density $p_{k_j}$ of \eqref{eq:gamma} for
$j=1,2$. By the Beta--Gamma algebra the ratio $B:=G_1/G_2$ is
$\mathrm{Beta}(k_1,m)$ distributed and is \emph{independent} of $G_2$; hence
\begin{equation}
\mathbb{E}\!\left[X_1\mid G_2\right]
=\frac{1}{k_1}\,\mathbb{E}\!\left[G_2B\mid G_2\right]
=\frac{G_2}{k_1}\,\mathbb{E}[B]
=\frac{G_2}{k_1}\cdot\frac{k_1}{k_2}
=X_2 .
\label{eq:condexp}
\end{equation}
Thus $X_2$ is a conditional expectation of $X_1$, which is precisely the
statement $X_2\le_{\mathrm{cx}}X_1$ in the convex order~\cite{ShakedShanthikumar}.
Since $\eta$ is convex, the conditional Jensen inequality applied to
\eqref{eq:condexp} gives $\eta(X_2)\le\mathbb{E}[\eta(X_1)\mid G_2]$, and taking
expectations, $\langle\eta\rangle_{k_2}\le\langle\eta\rangle_{k_1}$.

For strictness, note that $\eta$ is \emph{strictly} convex on $[0,\infty)$ and
that, conditionally on $G_2=g$, the variable $X_1=gB/k_1$ is non-degenerate,
because $\mathrm{Beta}(k_1,m)$ is non-degenerate for every $k_1,m>0$. The
conditional Jensen inequality is therefore strict almost surely, and the
inequality between the expectations is strict.
\end{proof}

\begin{remark}
The construction contains the elementary integer case as a special instance: for
$k$ integer, \eqref{eq:gamma} is the law of the mean of $k$ i.i.d.\ unit
exponentials, and \eqref{eq:condexp} reduces to the observation that the mean of
$k+1$ such variables is the average of the $k+1$ sample means obtained by
omitting one of them. The Beta--Gamma argument removes the restriction to
integers without lengthening the proof.
\end{remark}

\begin{corollary}[The value is not determined by the mean]
\label{cor:range}
As $k$ ranges over $(0,\infty)$ at fixed $\alpha$, $\langle\eta\rangle_k$ sweeps
the whole interval $\left(\,(1+\alpha)^{-1},\,1\,\right)$: the lower endpoint is
approached as $k\to\infty$, where $p_k\Rightarrow\delta(x-1)$, and the upper as
$k\to0$, where the prior places vanishing mass at large $x$ and
$\eta\to1$ in probability.
\end{corollary}

\noindent
Numerically, at $\alpha=1$ the mean efficiency is $0.9937$ at $k=10^{-3}$,
$0.7270$ at $k=0.25$, $0.5963$ at $k=1$, $0.5237$ at $k=5$ and $0.5001$ at
$k=10^3$, against a floor of $0.5000$ (Table~\ref{tab:sens},
Fig.~\ref{fig:2}).

\subsection*{3.4 Summary: a floor and nothing more}

Three statements survive this section, and it is worth setting them out before
saying what does not.

First, the bound \eqref{eq:jensen} is unconditional. It assumes no prior, no
maximum-entropy argument and no functional form --- only $S\ge0$ and a finite
mean --- and by Proposition~\ref{prop:large} the effect it describes grows as
$\ln\alpha$. Dissipation fluctuations always raise the mean efficiency, and
increasingly so the more dissipative the machine.

Second, under the information state of Assumption~3 --- mean known, nothing else
--- the maximum-entropy prior is unique. Equation~\eqref{eq:main} is therefore the unique
reference value consistent with that state: a well-defined benchmark rather than
an estimate, and one that another author starting from the same declared
information would be obliged to reproduce.

Third, \eqref{eq:gammares} converts whatever additional information one does hold,
expressed as a value of $\langle\ln x\rangle$, into the corresponding revised
benchmark. A reader who knows more than the mean is not obliged to accept $k=1$,
and Fig.~\ref{fig:2} is the conversion chart.

The care taken over the wording of those three is not incidental. It reflects a
boundary that Corollary~\ref{cor:range} makes exact, and which is better stated
plainly than left for the reader to derive: \emph{knowing the mean dissipation
constrains the mean efficiency only from below.} The bound \eqref{eq:jensen} is
saturated by deterministic dissipation, while the opposite extreme is limited by
nothing but $\eta\le1$, so the mean alone singles out no particular value. Any
number quoted for $\langle\eta\rangle$ --- including \eqref{eq:main} --- follows from the
information state assumed and not from the mean dissipation measured. This is a
property of the problem rather than of the present treatment: no choice of prior
could do better, because the mean genuinely does not contain the information.

Delimited this way the framework has a clear operational content. A measured
efficiency below $(1+\alpha)^{-1}$ contradicts the model outright, and that test
is prior-free. One lying above \eqref{eq:main} indicates dissipation more variable than a
mean-only description implies, and Fig.~\ref{fig:2} converts the excess into an
effective $k$ --- that is, into a statement about $\langle\ln x\rangle$, a second
measurable moment. The framework thus turns an efficiency measurement into
information about the shape of the dissipation distribution, which is a more
useful thing to extract from it than a single predicted number would have been.

\begin{table}[H]
\centering
\caption{$\langle\eta\rangle_k$ from \eqref{eq:gammares} at fixed mean
dissipation. The last column is the Jensen floor of
Proposition~\ref{prop:jensen}, approached as $k\to\infty$; as $k\to0$ every row
tends to $1$}
\label{tab:sens}
\begin{tabular}{lcccccc c}
\toprule
$\alpha$ & $k=0.25$ & $k=0.5$ & $k=1$ & $k=2$ & $k=5$ & $k=20$ & $(1+\alpha)^{-1}$\\
\midrule
$0.5$ & 0.8020 & 0.7579 & 0.7227 & 0.6985 & 0.6806 & 0.6703 & 0.6667\\
$1.0$ & 0.7270 & 0.6557 & 0.5963 & 0.5547 & 0.5237 & 0.5062 & 0.5000\\
$2.0$ & 0.6465 & 0.5456 & 0.4615 & 0.4037 & 0.3627 & 0.3407 & 0.3333\\
$3.0$ & 0.5990 & 0.4819 & 0.3856 & 0.3218 & 0.2788 & 0.2571 & 0.2500\\
\bottomrule
\end{tabular}
\end{table}

\begin{figure}[H]
\centering
\includegraphics{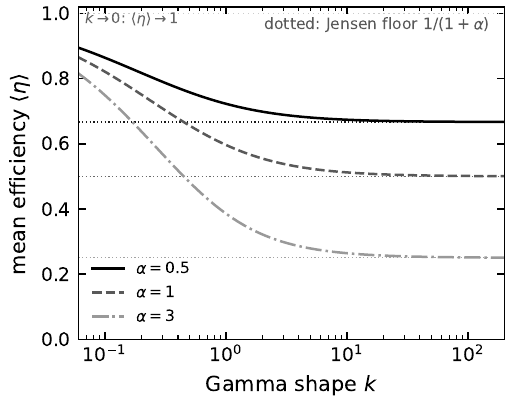}
\caption{Mean efficiency \eqref{eq:gammares} against the Gamma shape $k$ at fixed
mean dissipation, for three values of $\alpha$. Increasing $k$ makes the
dissipation more predictable at the same mean; the curves descend monotonically
(Proposition~\ref{prop:mono}) onto the Jensen floors $1/(1+\alpha)$ (fine dotted)
and rise towards $1$ as $k\to0$ (Corollary~\ref{cor:range})}
\label{fig:2}
\end{figure}

\section*{4. Two-sided bounds from two moments}

\subsection*{4.1 The remaining question}

Corollary~\ref{cor:range} is a negative result with a precise shape: the mean
dissipation bounds the mean efficiency from below and does nothing else. The
bound \eqref{eq:jensen} is sharp, and the opposite extreme is limited only by
$\eta\le1$. The obvious question is what a \emph{second} measurement supplies.
We take that second measurement to be the variance of the dissipation,
\begin{equation}
\sigma^2 \;\equiv\; \frac{\mathrm{Var}(S)}{\langle S\rangle^2}
\;=\;\mathrm{Var}(x),
\label{eq:sigma}
\end{equation}
which is directly accessible in any experiment that already resolves
$\langle S\rangle$ at the single-realisation level. The answer is that two
moments close the problem: they determine $\langle\eta\rangle$ to within an
interval that is sharp at both ends and that shrinks to a point as
$\sigma^2\to0$.

\subsection*{4.2 The upper bound}

\begin{proposition}[Upper bound from the variance]
\label{prop:upper}
Let $S\ge0$ have mean $\mu$ and relative variance $\sigma^2$, and let
$\alpha=T_0\mu/W$. Write $b=1+\sigma^2$ and $\beta=1+\alpha b$. Then
\begin{equation}
\langle\eta\rangle\;\le\;\mathcal{U}(\alpha,\sigma^2)
\;\equiv\;\frac{\sigma^2}{1+\sigma^2}
+\frac{1}{(1+\sigma^2)\left[1+\alpha(1+\sigma^2)\right]}
\;=\;1-\frac{\alpha}{\beta},
\label{eq:upper}
\end{equation}
and the bound is attained.
\end{proposition}

\begin{proof}
Consider the quadratic
\begin{equation}
q(x)\;=\;1-\frac{\alpha(2\beta-1)}{\beta^{2}}\,x+\frac{\alpha^{2}}{\beta^{2}}\,x^{2},
\label{eq:majorant}
\end{equation}
chosen so that $q(0)=\eta(0)$, $q(b)=\eta(b)$ and $q'(b)=\eta'(b)$. Since
$1+\alpha x>0$ on $x\ge0$, the sign of $q-\eta$ is that of
$q(x)(1+\alpha x)-1$, and a direct expansion gives the identity
\begin{equation}
q(x)\left(1+\alpha x\right)-1\;=\;\frac{\alpha^{3}}{\beta^{2}}\;x\,(x-b)^{2},
\label{eq:cubic}
\end{equation}
which is non-negative for every $x\ge0$. Hence $q\ge\eta$ pointwise on the
support, and
\[
\langle\eta\rangle\;\le\;\langle q\rangle
\;=\;1-\frac{\alpha(2\beta-1)}{\beta^{2}}+\frac{\alpha^{2}}{\beta^{2}}\,(1+\sigma^{2})
\;=\;1-\frac{\alpha}{\beta},
\]
using $\langle x\rangle=1$, $\langle x^{2}\rangle=1+\sigma^{2}$ and
$\alpha b=\beta-1$. The right-hand side depends on the distribution only through
the two prescribed moments. Equality requires $q=\eta$ almost surely, which by
\eqref{eq:cubic} confines the support to $\{0,b\}$; the two-point law
\begin{equation}
\Pr\left[x=0\right]=\frac{\sigma^{2}}{1+\sigma^{2}},
\qquad
\Pr\left[x=1+\sigma^{2}\right]=\frac{1}{1+\sigma^{2}}
\label{eq:extremal}
\end{equation}
has the required mean and variance and attains \eqref{eq:upper}.
\end{proof}

\noindent
The bound was checked independently by linear programming over a dense grid of
support points subject to the two moment constraints: the numerical optimum
agrees with \eqref{eq:upper} to $10^{-13}$ or better and the optimiser places all
mass at $0$ and $1+\sigma^{2}$, as \eqref{eq:extremal} requires.

\begin{remark}[Provenance of the method]
\label{rem:moment}
Proposition~\ref{prop:upper} is an instance of the classical Chebyshev--Markov
moment problem --- bounding $\langle f(X)\rangle$ over all distributions with
prescribed moments --- whose extremal solutions are discrete and whose optimal
bounds are certified by polynomial majorants of exactly the kind used above. The
general theory is set out by Karlin and Studden~\cite{KarlinStudden}; modern
treatments recast the problem as a semidefinite
program~\cite{Vandenberghe2007,Bertsimas2005}, and the linear-programming check
just described is the discretised form of that formulation. We give the short
self-contained proof because the majorant is explicit for
$\eta=(1+\alpha x)^{-1}$ and the argument is then two lines, but we claim no
novelty for the machinery. What is new here is its application to stochastic
efficiency, and the physical reading of the extremal law given next.
\end{remark}

\subsection*{4.3 The extremal machine}

Equation~\eqref{eq:extremal} has a direct physical reading. Among all machines
with a given mean and variance of dissipation, the one with the highest mean
efficiency is \emph{intermittently reversible}: with probability
$\sigma^{2}/(1+\sigma^{2})$ it dissipates nothing at all, and otherwise it
dissipates $(1+\sigma^{2})\langle S\rangle$, somewhat more than its average. No
distribution that spreads its dissipation smoothly can do as well. The mechanism
is the convexity already used in Proposition~\ref{prop:jensen}, pushed to its
limit: because $\eta$ saturates at $1$, moving probability towards zero
dissipation buys efficiency at a diminishing rate, so the optimal allocation puts
as much mass as the variance permits exactly at $S=0$.

This is a statement with experimental content, and the objection it invites ---
that a two-point law is an artifact of the moment constraints rather than a
machine anyone could build --- is answered by the majorant construction itself.
No machine sits exactly on the bound; what matters is what happens near it.

\begin{proposition}[Stability of the extremal law]
\label{prop:stability}
Let $x=S/\langle S\rangle$ have unit mean and relative variance $\sigma^{2}$, put
$b=1+\sigma^{2}$, $\beta=1+\alpha b$, and let
\begin{equation}
\psi(x)\;=\;\frac{\alpha^{3}}{\beta^{2}}\,\frac{x\,(x-b)^{2}}{1+\alpha x}\;\ge\;0,
\label{eq:psi}
\end{equation}
which vanishes precisely at $x=0$ and $x=b$. Then the shortfall from the bound is
exactly the mean of $\psi$,
\begin{equation}
\mathcal{U}(\alpha,\sigma^{2})-\langle\eta\rangle\;=\;\langle\psi\rangle ,
\label{eq:shortfall}
\end{equation}
and consequently, for every $\delta\in(0,b/2)$,
\begin{equation}
\Pr\!\left[\,\mathrm{dist}\!\left(x,\{0,b\}\right)>\delta\,\right]
\;\le\;\frac{\mathcal{U}(\alpha,\sigma^{2})-\langle\eta\rangle}{m(\delta)},
\qquad
m(\delta)=\min\!\left\{\psi(\delta),\,\psi(b-\delta),\,\psi(b+\delta)\right\}.
\label{eq:stability}
\end{equation}
\end{proposition}

\begin{proof}
Dividing the identity \eqref{eq:cubic} by $1+\alpha x>0$ gives
$q(x)-\eta(x)=\psi(x)$ pointwise on $x\ge0$; taking expectations and using
$\langle q\rangle=\mathcal{U}$, which depends only on the two prescribed moments,
yields \eqref{eq:shortfall}. Since $\psi\ge0$ and $\psi\ge m(\delta)$ on the set
where $\mathrm{dist}(x,\{0,b\})>\delta$, Markov's inequality gives
\eqref{eq:stability}. The stated form of $m(\delta)$ holds because $\psi$
vanishes at $0$ and $b$, is strictly positive in between and beyond, and is
increasing for $x>b$, so its minimum over the excluded set is attained at
$\delta$, $b-\delta$ or $b+\delta$. Expanding \eqref{eq:psi} about $x=b$ shows
$m(\delta)=\Theta(\delta^{2})$ as $\delta\to0$.
\end{proof}

\noindent
Equation~\eqref{eq:shortfall} rewards a slow reading: the deficiency of a machine
from the upper bound \emph{is} the average of a function that measures distance
from the two-point law. Intermittency is therefore neither assumed nor an
artifact --- it is forced by near-optimality, quantitatively, for every
distribution with the given moments. At $\alpha=\sigma^{2}=1$, for instance, a
machine within $10^{-3}$ of the bound has at most $6\%$ of its dissipation
further than $0.5$ from $\{0,\,2\langle S\rangle\}$.

The practical consequence is that the prediction can be tested without requiring
a device to attain the bound exactly. A machine measured to sit \emph{close} to
$\mathcal{U}(\alpha,\sigma^{2})$ must be operating intermittently, with a
substantial fraction of realisations near reversibility, and
\eqref{eq:stability} turns the measured shortfall into a quantitative statement
about the trajectory record --- one that can be checked independently of any
efficiency measurement.

\subsection*{4.4 The variance does not raise the floor}

\begin{proposition}[The lower bound is unchanged]
\label{prop:lowerstill}
For every $\sigma^{2}>0$ the infimum of $\langle\eta\rangle$ over distributions
on $[0,\infty)$ with mean $\mu$ and relative variance $\sigma^{2}$ equals
$(1+\alpha)^{-1}$, and is not attained.
\end{proposition}

\begin{proof}
The bound $\langle\eta\rangle>(1+\alpha)^{-1}$ is Proposition~\ref{prop:jensen},
strict because no admissible law is degenerate. For the converse take the
two-point family supported on $\{A,B\}$ with $0\le A<1<B$ and
$(1-A)(B-1)=\sigma^{2}$, which has the prescribed moments for every
$A\in[0,1)$. As $A\to1^{-}$ the mass at $A$ tends to $1$ while $B\to\infty$;
since $\eta$ is bounded, the far atom contributes nothing in the limit and
$\langle\eta\rangle\to\eta(1)=(1+\alpha)^{-1}$.
\end{proof}

\noindent
The asymmetry is worth noting. Knowing the variance restricts how much
probability can sit at zero dissipation, and so caps the efficiency; it does not
prevent a law from concentrating near the mean with a vanishing far tail, and so
does not raise the floor. A third moment would be required for that.

\subsection*{4.5 The admissible interval}

Propositions~\ref{prop:jensen} and~\ref{prop:upper} together give
\begin{equation}
\frac{1}{1+\alpha}\;<\;\langle\eta\rangle\;\le\;
\frac{\sigma^{2}}{1+\sigma^{2}}
+\frac{1}{(1+\sigma^{2})\left[1+\alpha(1+\sigma^{2})\right]},
\label{eq:interval}
\end{equation}
both ends sharp, from two measured numbers and no distributional assumption
whatever. The width behaves as
\begin{equation}
\mathcal{U}-\frac{1}{1+\alpha}\;=\;\sigma^{2}\,\frac{\alpha^{2}}{(1+\alpha)^{2}}
+O(\sigma^{4}),
\label{eq:width}
\end{equation}
verified numerically: the exact width is $\left[1+\alpha+\alpha\sigma^{2}\right]^{-1}
\left(1+\alpha\right)$ times \eqref{eq:width}, so at $\alpha=1$ the leading form
overstates it by a factor $1.0005$ at $\sigma^{2}=10^{-3}$ and $1.005$ at
$\sigma^{2}=10^{-2}$. Two limits confirm the picture. As $\sigma^{2}\to0$ the interval
closes onto the deterministic value: a machine with reproducible dissipation has
its mean efficiency fixed by the mean alone. As $\sigma^{2}\to\infty$ the upper
bound tends to $1$ and \eqref{eq:interval} degenerates to
Corollary~\ref{cor:range}, which is thus recovered as the limit in which the
second measurement carries no information.

It is also instructive to locate the maximum-entropy answer of Section~7 within
the interval. The exponential prior has $\sigma^{2}=1$, so its interval is the
$\sigma^{2}=1$ one, and the fraction of that interval which the maximum-entropy
value occupies is $0.98$ at $\alpha=0.01$, $0.58$ at $\alpha=1$ and $0.06$ at
$\alpha=100$: an interior point at every operating point, never an edge --- as it
should be, being a noncommittal choice rather than an extremal one. The location
is not universal. Within the Gamma family \eqref{eq:gamma}, whose members have
$\sigma^{2}=1/k$, the fraction tends to $(1+\alpha)^{-1}$ as the width closes,
the computed values at $\alpha=1$ being $0.5025$ at $k=50$ and $0.5002$ at
$k=500$; the exponential prior is the $k=1$ member and sits higher.

\begin{figure}[H]
\centering
\includegraphics{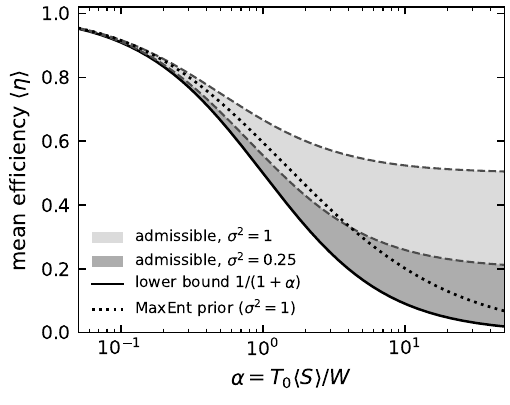}
\caption{The admissible interval \eqref{eq:interval} as a function of $\alpha$,
for two values of the relative variance $\sigma^{2}$ of the dissipation. The
lower edge is the prior-free bound of Proposition~\ref{prop:jensen} and is common
to both; the upper edges (dashed) are Proposition~\ref{prop:upper}. Shrinking
$\sigma^{2}$ closes the band onto the deterministic value. The dotted curve is
the maximum-entropy result \eqref{eq:main}, which has $\sigma^{2}=1$ and lies
inside the corresponding band at every $\alpha$}
\label{fig:3}
\end{figure}

\subsection*{4.6 What two measurements determine}

Section~3 established that one moment of the dissipation determines nothing about
the mean efficiency beyond a floor. Equation~\eqref{eq:interval} is the
corresponding positive statement: two moments determine it to within a computable
interval, sharp at both ends, whose width vanishes with the variance. The
practical procedure requires no prior and no dynamical model. Measure $W$,
$\langle S\rangle$ and $\mathrm{Var}(S)$; form $\alpha$ and $\sigma^{2}$;
\eqref{eq:interval} then brackets the mean efficiency. A measurement falling
outside the interval falsifies the energy accounting \eqref{eq:budget} or the
non-negativity of the coarse-grained dissipation, and a measurement near the
upper edge additionally predicts intermittent operation through
\eqref{eq:extremal}. One question is left open by this: whether $\alpha$ and
$\sigma^{2}$ are free, or whether the dynamics constrain the pair a machine can
present. Section~6.5 shows that they are constrained once the conditioning on $W$
is removed --- the interval's upper edge is capped and its width floored by the
precision of the delivered work alone.

\section*{5. A third moment tightens the floor}

\subsection*{5.1 Why the variance leaves the floor alone}

Proposition~\ref{prop:lowerstill} showed that the second moment caps the mean
efficiency without raising its floor, and its proof shows why. The family
saturating the infimum places an atom at $A\to1^{-}$ together with a vanishing
atom at $B\to\infty$, holding $(1-A)(B-1)=\sigma^{2}$; along it
\begin{equation}
\langle x^{3}\rangle\;\sim\;\sigma^{2}B\;\longrightarrow\;\infty .
\label{eq:m3div}
\end{equation}
The distributions that force the floor down are therefore precisely those with
unbounded third moment, and constraining $\langle x^{3}\rangle$ must exclude
them. This section carries that out, and the result is one-sided: the floor moves
with $m_{3}$ while the ceiling, which uses only two moments, does not.

\subsection*{5.2 The three-moment lower bound}

\begin{proposition}[Lower bound from three moments]
\label{prop:third}
Let $x=S/\langle S\rangle\ge0$ have $\langle x\rangle=1$,
$\mathrm{Var}(x)=\sigma^{2}>0$ and $\langle x^{3}\rangle=m_{3}$, the last
necessarily satisfying $m_{3}\ge(1+\sigma^{2})^{2}$ by \eqref{eq:m3min} below.
Put
\begin{equation}
D\;=\;\frac{m_{3}-1-3\sigma^{2}}{\sigma^{2}} .
\label{eq:Ddef}
\end{equation}
Then
\begin{equation}
\langle\eta\rangle\;\ge\;\mathcal{L}(\alpha,\sigma^{2},m_{3})
\;\equiv\;\frac{1+\alpha+\alpha D}
{(1+\alpha)^{2}+\alpha(1+\alpha)D-\alpha^{2}\sigma^{2}} ,
\label{eq:lower3}
\end{equation}
and the bound is attained.
\end{proposition}

\begin{proof}
Write a two-point law on $\{A,B\}$ as $A=1-u$, $B=1+v$ with $u,v>0$; the mean and
variance constraints are equivalent to $uv=\sigma^{2}$ with weights
$p=v/(u+v)$ at $A$ and $u/(u+v)$ at $B$. Expanding the third moment and using
$uv=\sigma^{2}$,
\[
\langle x^{3}\rangle=\frac{v(1-u)^{3}+u(1+v)^{3}}{u+v}
=1+3\sigma^{2}+\sigma^{2}(v-u),
\]
so the three constraints are equivalent to $uv=\sigma^{2}$ together with
$v-u=D$, giving
\begin{equation}
u=\tfrac12\!\left(\sqrt{D^{2}+4\sigma^{2}}-D\right),
\qquad
v=\tfrac12\!\left(\sqrt{D^{2}+4\sigma^{2}}+D\right).
\label{eq:uv}
\end{equation}

Let $q$ be the cubic Hermite interpolant of $\eta$ at the nodes $\{A,A,B,B\}$,
so that $q$ and $\eta$ agree in value and slope at both points. Then
$q-\eta$ has double zeros at $A$ and $B$, and since $1+\alpha x>0$ on $x\ge0$,
\begin{equation}
q(x)\left(1+\alpha x\right)-1\;=\;c_{3}\,\alpha\,(x-A)^{2}(x-B)^{2},
\label{eq:quartic}
\end{equation}
$c_{3}$ being the leading coefficient of $q$. For $\eta(x)=(1+\alpha x)^{-1}$ the
relevant divided difference is
\begin{equation}
c_{3}\;=\;\eta[A,A,B,B]\;=\;\frac{-\alpha^{3}}{(1+\alpha A)^{2}(1+\alpha B)^{2}}\;<\;0,
\label{eq:c3}
\end{equation}
so the right-hand side of \eqref{eq:quartic} is non-positive for every $x\ge0$
and $q\le\eta$ there. Hence
$\langle\eta\rangle\ge\langle q\rangle
=c_{0}+c_{1}+c_{2}(1+\sigma^{2})+c_{3}m_{3}$, which depends on the distribution
only through the three prescribed moments. Evaluating $\langle q\rangle$ at the
two-point law and simplifying with \eqref{eq:uv} gives \eqref{eq:lower3};
equality holds there, since \eqref{eq:quartic} confines the support of any
minimiser to $\{A,B\}$.
\end{proof}

\noindent
Proposition~\ref{prop:third} is the next order of the construction of
Remark~\ref{rem:moment}: the majorant of Proposition~\ref{prop:upper} becomes a
cubic minorant with double contact at both extremal points, and the quartic
\eqref{eq:quartic} replaces the cubic \eqref{eq:cubic}. Again the classical
theory~\cite{KarlinStudden} guarantees the two-point structure in advance; the
work here is evaluating it for this $\eta$.

Expression \eqref{eq:lower3} contains no radical: the square roots in
\eqref{eq:uv} cancel in the combination that appears. Checked against linear
programming over a dense grid with the three moment constraints imposed, it
agrees to the grid resolution ($\sim10^{-7}$) at every parameter set tested. At
$\alpha=\sigma^{2}=1$ it reduces to the rational family
$\mathcal{L}=(m_{3}-2)/(2m_{3}-5)$.

\subsection*{5.3 Both ends of the bracket}

Two limits fix the meaning of \eqref{eq:lower3}.

As $m_{3}\to\infty$ we have $D\to\infty$ and
$\mathcal{L}\to(1+\alpha)^{-1}$: the Jensen floor is recovered, as it must be,
since an unconstrained third moment readmits the escaping family
\eqref{eq:m3div}.

At the other end there is a floor on $m_{3}$ itself, and it holds for every law
on $[0,\infty)$ rather than only for the two-point family. By Cauchy--Schwarz,
$\langle x^{2}\rangle=\langle x^{1/2}\cdot x^{3/2}\rangle
\le\langle x\rangle^{1/2}\langle x^{3}\rangle^{1/2}$, so with $\langle x\rangle=1$
the third moment cannot fall below
\begin{equation}
m_{3}^{\min}=\left\langle x^{2}\right\rangle^{2}=\left(1+\sigma^{2}\right)^{2},
\label{eq:m3min}
\end{equation}
with equality exactly when $x^{1/2}$ and $x^{3/2}$ are proportional in
$L^{2}$, that is when $x$ is supported on at most one non-zero point. Within the
extremal family \eqref{eq:uv} the same value is reached at $u=1$, where $A=0$ ---
and the law there is exactly the extremal \eqref{eq:extremal} of
Proposition~\ref{prop:upper}. Consequently
\begin{equation}
\mathcal{L}\!\left(\alpha,\sigma^{2},m_{3}^{\min}\right)
=\mathcal{U}\!\left(\alpha,\sigma^{2}\right),
\end{equation}
which we have verified to machine precision: at $\alpha=1$, $\sigma^{2}=1$ both
equal $2/3$; at $\alpha=3$, $\sigma^{2}=1/2$ both equal $5/11$. The bracket
therefore \emph{closes to a point} at the smallest third moment compatible with
the first two, and widens monotonically to \eqref{eq:jensen} as $m_{3}$ grows
(Fig.~\ref{fig:4}). Three moments determine
\begin{equation}
\mathcal{L}\!\left(\alpha,\sigma^{2},m_{3}\right)\;\le\;\langle\eta\rangle
\;\le\;\mathcal{U}\!\left(\alpha,\sigma^{2}\right),
\label{eq:bracket3}
\end{equation}
with the lower end sharp and attained by the explicit two-point law
\eqref{eq:uv}.

The upper end requires a word of care, because it is inherited from
Proposition~\ref{prop:upper} and uses only two moments. It remains valid, but it
is no longer attained once $m_{3}$ is prescribed above its minimum. By
Proposition~\ref{prop:stability} the shortfall $\mathcal{U}-\langle\eta\rangle$
equals $\langle\psi\rangle$ with $\psi$ vanishing only at $x\in\{0,b\}$, so
equality forces the law onto that support; with the first two moments fixed such
a law is unique --- it is \eqref{eq:extremal} --- and its third moment is exactly
$m_{3}^{\min}$. Hence $\langle\eta\rangle<\mathcal{U}$ strictly whenever
$m_{3}>m_{3}^{\min}$, and linear programming over the three-moment problem puts
the shortfall in the range $10^{-3}$ to $10^{-2}$ across the parameter sets we
tested. Only the floor in \eqref{eq:bracket3} therefore responds to the third
moment; sharpening the ceiling would require the upper solution of the
three-moment problem, which we do not compute here. What is sharp at both ends is
the two-moment interval \eqref{eq:interval}, and the collapse
$\mathcal{L}(\alpha,\sigma^{2},m_{3}^{\min})=\mathcal{U}(\alpha,\sigma^{2})$ is
the statement that at $m_{3}^{\min}$ the two coincide.

\begin{figure}[H]
\centering
\includegraphics{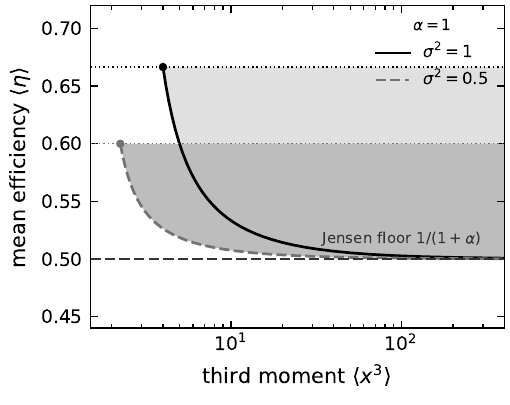}
\caption{The bracket \eqref{eq:bracket3} at $\alpha=1$ as a function of the third
moment, for two values of $\sigma^{2}$. Upper edges (fine dotted) are
$\mathcal{U}(\alpha,\sigma^{2})$, independent of $m_{3}$; lower edges are
$\mathcal{L}$ from \eqref{eq:lower3}. Circles mark $m_{3}^{\min}=(1+\sigma^{2})^{2}$,
where the bracket closes to a point. The long-dashed line is the Jensen floor,
approached as $m_{3}\to\infty$}
\label{fig:4}
\end{figure}

\subsection*{5.4 What the hierarchy shows, and what it costs}

The three results now form a sequence with a clear structure. One moment gives a
floor and nothing else (Proposition~\ref{prop:jensen} and
Corollary~\ref{cor:range}). A second moment adds a sharp ceiling but cannot move
the floor (Propositions~\ref{prop:upper} and~\ref{prop:lowerstill}). A third
moment moves the floor, and moves it all the way to the ceiling at the smallest
value it may take. Each is attained by a two-point law, and each follows from the
same construction: a polynomial that touches $\eta$ at the extremal support and
lies on one side of it everywhere on $[0,\infty)$, so that its expectation is
computable from the prescribed moments alone.

A caveat should accompany \eqref{eq:lower3} rather than follow it. The relative
variance $\sigma^{2}$ is accessible in any experiment that already resolves the
dissipation realisation by realisation; the third moment is not, requiring
substantially more statistics and being correspondingly noisier. The practical
content of this paper therefore remains \eqref{eq:interval}, which needs two
measured numbers. Proposition~\ref{prop:third} is included because it settles the
structural question left open by Proposition~\ref{prop:lowerstill} --- whether
the floor is movable at all, and by what --- and because the collapse at
$m_{3}^{\min}$ shows that the moment sequence closes rather than merely
improving indefinitely.

\section*{6. When the delivered work also fluctuates}

Assumption~\ref{as:task} conditions on $W$. This section removes that
conditioning and asks what survives. Two things do: the existence of the mean
efficiency, which becomes easier rather than harder, and the bounds themselves,
which hold verbatim once the correct dimensionless group is identified. What does
not survive is the interpretation of that group as a ratio of separately measured
averages.

\subsection*{6.1 Existence is not the obstacle}

\begin{proposition}[The joint mean always exists]
\label{prop:jointexist}
Let $W\ge0$ and $S\ge0$ be jointly distributed, not both zero almost surely. Then
$\eta=W/(W+T_0S)\in[0,1]$ pointwise, and $\langle\eta\rangle$ exists.
\end{proposition}

\begin{proof}
$T_0S\ge0$ gives $0\le W\le W+T_0S$, so the ratio lies in $[0,1]$ wherever the
denominator is positive; a bounded measurable function has a finite mean.
\end{proof}

\noindent
This is worth stating because the natural expectation is the opposite. The
divergence of Proposition~\ref{prop:exist} arises from admitting $S<0$, which
allows the denominator to vanish while the numerator does not; it has nothing to
do with a fluctuating numerator. When $W$ fluctuates but both variables remain
non-negative, numerator and denominator vanish together and the ratio stays
bounded. Relaxing Assumption~\ref{as:task} therefore costs nothing on the
existence question.

\subsection*{6.2 Two means determine nothing at all}

\begin{proposition}[No content in the marginal means]
\label{prop:jointfree}
Fix $\langle W\rangle=\bar W>0$ and $\langle S\rangle=\bar S>0$. Then
$\langle\eta\rangle$ may take any value in the open interval $(0,1)$, and both
endpoints are approached but not attained.
\end{proposition}

\begin{proof}
For the supremum take mass $1-\varepsilon$ at $\left(\bar W/(1-\varepsilon),\,0\right)$,
where $\eta=1$, and mass $\varepsilon$ at $\left(0,\,\bar S/\varepsilon\right)$,
where $\eta=0$; the marginal means are $\bar W$ and $\bar S$ and
$\langle\eta\rangle=1-\varepsilon$. For the infimum exchange the roles, placing
mass $1-\varepsilon$ at $\left(0,\,\bar S/(1-\varepsilon)\right)$ and mass
$\varepsilon$ at $\left(\bar W/\varepsilon,\,0\right)$, giving
$\langle\eta\rangle=\varepsilon$.
\end{proof}

\noindent
Proposition~\ref{prop:jointfree} is a sharper negative result than
Corollary~\ref{cor:range}: with $W$ conditioned on, one moment of $S$ at least
supplies a floor, whereas the two marginal means of an unconditioned $(W,S)$
supply nothing whatever. The reason is that $\eta$ is not a function of $W$ and
$S$ separately in any way those means control --- a point the next subsection
makes precise and, in doing so, repairs.

\subsection*{6.3 The reduction}

Write
\begin{equation}
R\;\equiv\;\frac{T_0S}{W},
\qquad\text{so that}\qquad
\eta\;=\;\frac{W}{W+T_0S}\;=\;\frac{1}{1+R}.
\label{eq:Rdef}
\end{equation}
The efficiency depends on the pair $(W,S)$ only through the single non-negative
random variable $R$, and \eqref{eq:Rdef} is algebraically identical to
\eqref{eq:etax} with $\alpha x$ replaced by $R$. Every result of Sections~4--6
therefore transfers without modification, on replacing the moments of $x=S/\mu$
by those of $R$. Writing $\rho=\langle R\rangle$ and
$\varsigma^{2}=\mathrm{Var}(R)/\rho^{2}$:

\begin{proposition}[Joint bounds]
\label{prop:jointbounds}
For any joint law of $(W,S)$ on $[0,\infty)^{2}$ with $\rho<\infty$,
\begin{equation}
\frac{1}{1+\rho}\;<\;\langle\eta\rangle\;\le\;
\frac{\varsigma^{2}}{1+\varsigma^{2}}
+\frac{1}{(1+\varsigma^{2})\left[1+\rho(1+\varsigma^{2})\right]},
\label{eq:jointinterval}
\end{equation}
with both ends sharp, and the three-moment refinement of
Proposition~\ref{prop:third} holds with $(\alpha,\sigma^{2},m_{3})$ replaced by
the corresponding moments of $R$.
\end{proposition}

\begin{proof}
Immediate from \eqref{eq:Rdef} and
Propositions~\ref{prop:jensen}, \ref{prop:upper} and~\ref{prop:third}, all of
which require only that the argument be a non-negative random variable of the
stated moments.
\end{proof}

\noindent
Two conditions attach. First $\rho<\infty$ is a genuine restriction: for
independent $W$ and $S$, $\rho=T_0\langle S\rangle\langle W^{-1}\rangle$, which
diverges whenever the density of $W$ is bounded away from zero at the origin ---
for $W$ exponentially distributed, for instance. The existence problem has not
disappeared but moved: it is now a condition on $R$ rather than on $\eta$, and
where it fails \eqref{eq:jointinterval} remains true but vacuous, the floor
collapsing to zero. Second, $\varsigma^{2}$ is the relative variance of the
ratio, not of the dissipation.

\subsection*{6.4 Why the ratio of means will not do}

The reduction is exact, but it changes what must be measured, and the
substitution a reader is most likely to make is unsafe. By Jensen's inequality,
for independent $W$ and $S$,
\begin{equation}
\rho\;=\;T_0\langle S\rangle\left\langle W^{-1}\right\rangle
\;\ge\;\frac{T_0\langle S\rangle}{\langle W\rangle}\;\equiv\;\alpha_{\text{naive}},
\label{eq:naive}
\end{equation}
so the ratio of means understates the dissipation ratio; and since the floor
$(1+\rho)^{-1}$ decreases in $\rho$, using $\alpha_{\text{naive}}$ produces a
floor that is too \emph{high}. It is therefore not a valid bound. This is not a
remote possibility. Among $6\times10^{4}$ randomly generated two-point joint laws
we found that very nearly half violate the bound computed from
$\alpha_{\text{naive}}$; one explicit case places mass $0.9$ at
$(W,T_0S)=(0.2,0.6)$ and mass $0.1$ at $(9,0.1)$, for which
$\alpha_{\text{naive}}=0.509$ and the resulting floor $0.663$ exceeds the true
mean efficiency $0.324$ by a factor of two.

Correlation only sharpens the warning, since \eqref{eq:naive} need not hold at
all when $W$ and $S$ are dependent, and $\rho$ may then fall on either side of
$\alpha_{\text{naive}}$. The moment $\rho$ must be formed from the ratio
realisation by realisation.

\subsection*{6.5 A precision--efficiency bound}

The reduction of Section~6.3 does more than preserve the bounds. It puts them in
a form to which a dynamical constraint can be attached, because the only place
the machine's dynamics needs to enter is through the single number $\rho$, and
$\rho$ is bounded below by the precision of the delivered work. The result is a
ceiling on the mean efficiency expressed entirely in fluctuation data.

The dynamical input is the uncertainty relation applied not to the dissipation
but to the work current. Applying it to the dissipation would not do: the relation
governs time-antisymmetric currents, and a dissipation obeying
Assumption~\ref{as:pos} cannot be one, since $X\ge0$ together with
$X\circ\Theta=-X$ under a reversal $\Theta$ with respect to which the path measure
is mutually absolutely continuous forces $X\equiv0$. The work current is subject
to no such obstruction.

\begin{assumption}[Work current in a steady state]
\label{as:wcurrent}
Over the observation window $[0,\mathcal{T}]$ the underlying dynamics is a
time-homogeneous Markov jump process on a finite state space, or an overdamped
diffusion, in a nonequilibrium steady state, or is driven periodically by a
time-symmetric protocol $\lambda(t)=\lambda(\mathcal{T}-t)$; the delivered work
$W$ is a time-integrated, time-antisymmetric current of that dynamics; and the
exergy accounting captures a fraction $r\in(0,1]$ of the total entropy
production, $\langle S\rangle=r\,k_B\langle\Sigma_{\rm tot}\rangle$.
\end{assumption}

\noindent
Write
\begin{equation}
\epsilon_{w}^{2}\;:=\;\frac{\mathrm{Var}(W)}{\langle W\rangle^{2}},
\qquad
\gamma\;:=\;\frac{2k_BT_0}{\langle W\rangle},
\label{eq:precdef}
\end{equation}
so that $\epsilon_{w}^{2}$ is the relative variance of the delivered work --- the
precision of the machine's output --- and $\gamma$ the thermal energy measured
against that output, the same small parameter as $\varepsilon$ of
the reciprocal of the work delivered per window in units of $k_BT_0$, doubled. Under Assumption~\ref{as:wcurrent} the
uncertainty relation~\cite{Barato2015,Gingrich2016,Horowitz2017} reads
$\epsilon_{w}^{2}\ge2/\langle\Sigma_{\rm tot}\rangle$, and substituting
$\langle\Sigma_{\rm tot}\rangle=\langle S\rangle/(rk_B)$ together with
$\alpha_{\text{naive}}=T_0\langle S\rangle/\langle W\rangle$ gives
\begin{equation}
\alpha_{\text{naive}}\;\ge\;\frac{r\gamma}{\epsilon_{w}^{2}} .
\label{eq:turnaive}
\end{equation}
A machine whose output is precise must dissipate, and \eqref{eq:turnaive} says how
much, in units of the work it delivers.

What is needed is the same statement for $\rho$, and Section~6.4 warns that
$\rho$ and $\alpha_{\text{naive}}$ are different numbers. The gap is an identity.

\begin{lemma}[When the ratio of means understates the dissipation ratio]
\label{lem:cov}
With $R=T_0S/W$ and $\rho=\langle R\rangle$,
\begin{equation}
\alpha_{\text{naive}}\;=\;\rho+\frac{\mathrm{Cov}(R,W)}{\langle W\rangle},
\qquad\text{so}\qquad
\rho\ \ge\ \alpha_{\text{naive}}
\iff
\mathrm{Cov}(R,W)\le0 .
\label{eq:covcrit}
\end{equation}
In particular $\mathrm{Cov}(R,W)\le0$ whenever $W$ and $S$ are independent, and
$\mathrm{Cov}(R,W)=0$ when $S$ is proportional to $W$.
\end{lemma}

\begin{proof}
$T_0S=RW$, so $T_0\langle S\rangle=\langle RW\rangle=\rho\langle W\rangle
+\mathrm{Cov}(R,W)$; divide by $\langle W\rangle$. For independent $W$ and $S$,
$\mathrm{Cov}(R,W)=T_0\langle S\rangle\left(1-\langle W^{-1}\rangle\langle
W\rangle\right)\le0$ by Cauchy--Schwarz. If $S=cW$ then $R$ is constant.
\end{proof}

The criterion is physically transparent: $\mathrm{Cov}(R,W)\le0$ says the
dissipation-to-work ratio does not rise with the work delivered, so that a
realisation in which the machine does more is not a realisation in which it is
proportionally less efficient. It holds with equality for a machine whose
dissipation tracks its output, and strictly for one whose fluctuations in the two
are independent. It fails when $S$ grows superlinearly in $W$.

\begin{proposition}[Precision--efficiency bound]
\label{prop:pareto}
Under Assumption~\ref{as:wcurrent}, and provided $\mathrm{Cov}(R,W)\le0$,
\begin{equation}
\langle\eta\rangle\;\le\;
\Phi\!\left(\epsilon_{w}^{2},\varsigma^{2};r\gamma\right)
\;:=\;\frac{\varsigma^{2}}{1+\varsigma^{2}}
+\frac{\epsilon_{w}^{2}}
{(1+\varsigma^{2})\left[\epsilon_{w}^{2}+r\gamma(1+\varsigma^{2})\right]} .
\label{eq:pareto}
\end{equation}
\end{proposition}

\begin{proof}
The upper bound of Proposition~\ref{prop:jointbounds} gives
$\langle\eta\rangle\le\mathcal{U}(\rho,\varsigma^{2})$ with
$\mathcal{U}$ as in \eqref{eq:upper}. Differentiating,
\begin{equation}
\frac{\partial\mathcal{U}}{\partial\rho}
=-\frac{1}{\left[1+\rho(1+\varsigma^{2})\right]^{2}}<0
\label{eq:Umono}
\end{equation}
for every $\rho>0$ and $\varsigma^{2}\ge0$, so $\mathcal{U}$ is strictly
decreasing in its first argument and any lower bound on $\rho$ may be substituted
for it. Lemma~\ref{lem:cov} and \eqref{eq:turnaive} give
$\rho\ge\alpha_{\text{naive}}\ge r\gamma/\epsilon_{w}^{2}$; substituting
$\rho=r\gamma/\epsilon_{w}^{2}$ into \eqref{eq:jointinterval} and multiplying the
second term above and below by $\epsilon_{w}^{2}$ yields \eqref{eq:pareto}.
\end{proof}

Three features of \eqref{eq:pareto} are worth separating.

\emph{It is a trade-off, and in the expected direction.} $\Phi$ is strictly
increasing in $\epsilon_{w}^{2}$, with
$\partial\Phi/\partial\epsilon_{w}^{2}=r\gamma
\left[\epsilon_{w}^{2}+r\gamma(1+\varsigma^{2})\right]^{-2}>0$. A machine
whose power output is more reproducible has a strictly lower ceiling on its mean
efficiency. As $\epsilon_{w}^{2}\to\infty$ the constraint disappears,
$\Phi\to1$; as $\epsilon_{w}^{2}\to0$ it bites hardest,
$\Phi\to\varsigma^{2}/(1+\varsigma^{2})$, and a machine with perfectly
reproducible output is held below the mass its dissipation law places at zero.
Equation \eqref{eq:pareto} is thus a Pareto frontier in the plane of output
precision and mean efficiency, with the dissipation's relative variance as the
family parameter.

\emph{At $\varsigma^{2}=0$ it is exactly the bound of Pietzonka, Barato and
Seifert.} Setting the dissipation's relative variance to zero,
\begin{equation}
\Phi\!\left(\epsilon_{w}^{2},0;r\gamma\right)
=\frac{\epsilon_{w}^{2}}{\epsilon_{w}^{2}+r\gamma}
=\left[1+\frac{r\gamma}{\epsilon_{w}^{2}}\right]^{-1},
\label{eq:pbslimit}
\end{equation}
and for a molecular motor with $r=1$, mean velocity $v$, displacement diffusion
coefficient $D$ and external force $f$ observed over time $t$, one has
$\epsilon_{w}^{2}=2D/(v^{2}t)$ and $\gamma=2k_BT_0/(fvt)$, whence
$\gamma/\epsilon_{w}^{2}=vk_BT_0/(Df)$ and \eqref{eq:pbslimit} becomes
$\left(1+vk_BT_0/Df\right)^{-1}$, the bound of~\cite{Pietzonka2016}
(Section~8.6). That bound constrains the ratio of means
$\langle W\rangle/(\langle W\rangle+T_0\langle S\rangle)$, which
Proposition~\ref{prop:jointfree} and Section~6.4 show is not the mean efficiency.

\emph{The extension is strict, and the direction matters.} Since $\Phi$ is
strictly increasing in $\varsigma^{2}$, with
$\partial\Phi/\partial\varsigma^{2}=(r\gamma)^{2}
\left[\epsilon_{w}^{2}+r\gamma(1+\varsigma^{2})\right]^{-2}>0$, the ceiling on
$\langle\eta\rangle$ lies strictly above \eqref{eq:pbslimit} whenever the
dissipation fluctuates. Applying the ratio-of-means bound to the mean of the
ratio is therefore not merely inexact but unsafe in the same way
$\alpha_{\text{naive}}$ was in Section~6.4: it can be violated by a real machine.
Equation \eqref{eq:pareto} is the repair, and it costs one further measured
number, $\varsigma^{2}$.

\begin{figure}[H]
\centering
\includegraphics{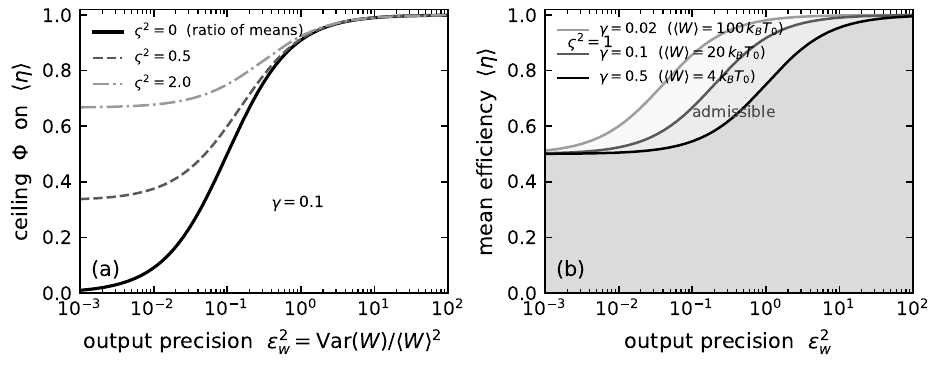}
\caption{The precision--efficiency frontier \eqref{eq:pareto} at $r=1$. (a)
$\Phi$ against the relative variance $\epsilon_{w}^{2}$ of the delivered work,
for three values of the dissipation's relative variance $\varsigma^{2}$, at
$\gamma=0.1$. The solid curve $\varsigma^{2}=0$ is the bound
of~\cite{Pietzonka2016}, which constrains the ratio of means; the mean efficiency
is bounded by the higher curves, and applying the $\varsigma^{2}=0$ curve to it
is unsafe. (b) The admissible region in the $(\epsilon_{w}^{2},\langle\eta\rangle)$
plane at $\varsigma^{2}=1$ for three values of $\gamma=2k_BT_0/\langle W\rangle$:
a smaller machine is pushed further from the top-left corner, where precise output
and high efficiency would coexist}
\label{fig:tur}
\end{figure}

\begin{corollary}[The interval cannot be closed either]
\label{cor:widthfloor}
Under the hypotheses of Proposition~\ref{prop:pareto}, the width
$\Delta:=\mathcal{U}-(1+\rho)^{-1}$ of the interval \eqref{eq:jointinterval}
satisfies
\begin{equation}
\Delta\;=\;\frac{\rho^{2}\varsigma^{2}}
{(1+\rho)\left(1+\rho+\rho\varsigma^{2}\right)}
\;\ge\;\frac{(r\gamma)^{2}\varsigma^{2}}
{\left[\epsilon_{w}^{2}+r\gamma\right]
 \left[\epsilon_{w}^{2}+r\gamma(1+\varsigma^{2})\right]} .
\label{eq:widthfloor}
\end{equation}
\end{corollary}

\begin{proof}
The identity follows from $\mathcal{U}=1-\rho/[1+\rho(1+\varsigma^{2})]$ together
with $1+\rho(1+\varsigma^{2})-1-\rho=\rho\varsigma^{2}$. Differentiating,
\[
\frac{\partial\Delta}{\partial\rho}
=\frac{\rho\varsigma^{2}\left(\rho\varsigma^{2}+2\rho+2\right)}
{(1+\rho)^{2}\left(1+\rho+\rho\varsigma^{2}\right)^{2}}\;>\;0
\]
for $\rho,\varsigma^{2}>0$, so $\Delta$ is strictly increasing in $\rho$ and the
bound $\rho\ge r\gamma/\epsilon_{w}^{2}$ established in the proof of
Proposition~\ref{prop:pareto} may be substituted for it.
\end{proof}

\noindent
One inequality therefore does both jobs: it caps the interval's upper edge and
floors its width, the two statements being the monotonicity of $\mathcal{U}$ and
of $\Delta$ in the same variable. Equation \eqref{eq:widthfloor} answers the
question left at the end of Section~4.6. The interval of \eqref{eq:interval} is
not an artefact of the two-moment information state, because the dynamics forbid
the $\rho\to0$ limit in which it closes: a machine cannot have both a precise
output and a negligible dissipation ratio. As $\epsilon_{w}^{2}\to0$ the floor
rises to $\varsigma^{2}/(1+\varsigma^{2})$, the whole of the interval; as
$\epsilon_{w}^{2}\to\infty$ it falls to zero, the constraint disappearing together
with the ceiling. It vanishes at $\varsigma^{2}=0$, as it must: a machine with
reproducible dissipation has its mean efficiency fixed by \eqref{eq:Rdef}
whatever its output precision.

\subsection*{6.6 Consequence for the paper}

Assumption~\ref{as:task} is therefore not a restriction on the validity of the
bounds. It is a device for \emph{expressing} the moments of $R$: when $W$ is
fixed by the task, $R=\alpha x$ and the moments of $R$ are those of the
dissipation, rescaled. When $W$ fluctuates, the bounds are unchanged and the same
moments are required, but they must be measured on the ratio rather than
assembled from separate averages of $S$ and $W$. Everything the paper claims
about $\langle\eta\rangle$ thus holds for a general machine; what conditioning on
$W$ buys is not correctness but access to the inputs.

\section*{7. A maximum-entropy benchmark inside the interval}

Sections~3--6 bound the mean efficiency without committing to a distribution.
This section does the complementary calculation: it evaluates
$\langle\eta\rangle$ for the one distribution singled out by the information
state of Assumption~\ref{as:mean}, the maximum-entropy prior \eqref{eq:prior}.
The result is a closed form, it is an interior point of the bracket
\eqref{eq:interval} rather than an edge of it, and it is a benchmark rather than
a prediction --- a reference value against which a measurement can be read, in
the spirit of the discussion in Section~3.4.

\subsection*{7.1 The mean efficiency}

With the efficiency \eqref{eq:etax} and the maximum-entropy prior
\eqref{eq:prior} of Section~2, the mean efficiency is a one-dimensional integral
that evaluates in closed form.

\begin{proposition}
\label{prop:main}
For $\alpha>0$,
\begin{equation}
\langle\eta\rangle(\alpha)
\;=\;\int_0^\infty\frac{e^{-x}}{1+\alpha x}\,\mathrm{d}x
\;=\;\frac{1}{\alpha}\,e^{1/\alpha}\,E_1\!\left(\frac{1}{\alpha}\right),
\label{eq:main}
\end{equation}
where $E_1(z)=\int_z^\infty t^{-1}e^{-t}\,\mathrm{d}t$.
\end{proposition}

\begin{proof}
Substituting $t=x+1/\alpha$, so that $1+\alpha x=\alpha t$ and
$\mathrm{d}x=\mathrm{d}t$,
\[
\int_0^\infty\frac{e^{-x}}{1+\alpha x}\,\mathrm{d}x
=\int_{1/\alpha}^\infty\frac{e^{-(t-1/\alpha)}}{\alpha t}\,\mathrm{d}t
=\frac{e^{1/\alpha}}{\alpha}\int_{1/\alpha}^\infty\frac{e^{-t}}{t}\,\mathrm{d}t. \qedhere
\]
\end{proof}

\noindent
The integral in \eqref{eq:main} converges for every $\alpha>0$ because the
integrand is bounded by $e^{-x}$ on the domain of integration --- which is
precisely the content of Assumption~\ref{as:pos} and Proposition~\ref{prop:exist}: on $x\ge0$ the
denominator never approaches zero. The natural comparison throughout what
follows is the \emph{deterministic} value
\begin{equation}
\eta_{\mathrm{det}} \;\equiv\; \eta\!\left(\langle S\rangle\right)
\;=\;\frac{1}{1+\alpha},
\label{eq:det}
\end{equation}
the efficiency a machine would have if it dissipated its mean amount every time.
Proposition~\ref{prop:jensen} shows $\langle\eta\rangle>\eta_{\mathrm{det}}$ always; the
asymptotics quantify by how much.

\subsection*{7.2 The near-reversible limit}

\begin{proposition}
\label{prop:small}
As $\alpha\to0^+$,
\begin{equation}
\langle\eta\rangle(\alpha)\;\sim\;\sum_{n\ge0}(-1)^n\,n!\,\alpha^n
\;=\;1-\alpha+2\alpha^2-6\alpha^3+\cdots,
\label{eq:small}
\end{equation}
an asymptotic series with zero radius of convergence. Consequently
\begin{equation}
\langle\eta\rangle-\eta_{\mathrm{det}}\;=\;\alpha^2+O(\alpha^3).
\label{eq:gap}
\end{equation}
\end{proposition}

\begin{proof}
Insert $E_1(z)\sim e^{-z}\sum_{n\ge0}(-1)^n n!\,z^{-(n+1)}$ as $z\to\infty$ into
\eqref{eq:main} with $z=1/\alpha$; the factor $e^{1/\alpha}$ cancels $e^{-z}$ and
one power of $\alpha$ cancels the prefactor. Equation~\eqref{eq:gap} follows on
subtracting $\eta_{\mathrm{det}}=1-\alpha+\alpha^2-\alpha^3+\cdots$ term by term.
\end{proof}

\noindent
Two features of \eqref{eq:small} are worth drawing out. First, the leading two
terms coincide with those of $\eta_{\mathrm{det}}$: to first order in the
dissipation, fluctuations in $S$ are invisible, and a measurement accurate only
to $O(\alpha)$ cannot distinguish a fluctuating machine from a deterministic one
with the same mean dissipation. The distinction first appears at second order,
where \eqref{eq:gap} gives a clean prediction with no free parameter. Numerically
the ratio $(\langle\eta\rangle-\eta_{\mathrm{det}})/\alpha^2$ equals $0.952$ at
$\alpha=10^{-2}$ and $0.908$ at $\alpha=2\times10^{-2}$, approaching unity as
$\alpha\to0$.

Second, the series diverges. This is not a defect of the calculation but a
signature of the pole discussed in Section~2: expanding $(1+\alpha x)^{-1}$ in
powers of $\alpha$ and integrating term by term is illegitimate because the
expansion fails for $x>1/\alpha$, a region carrying probability $e^{-1/\alpha}$
under the prior. The series is therefore useful only up to optimal truncation,
at $n\approx1/\alpha$ terms, with an irreducible error of order $e^{-1/\alpha}$.
We confirmed this numerically: for $\alpha=0.1$ the best partial sum is at $n=9$
with error $1.8\times10^{-4}$, against $e^{-10}=4.5\times10^{-5}$; for
$\alpha=0.05$, at $n=19$ with error $1.1\times10^{-8}$ against
$e^{-20}=2.1\times10^{-9}$. The exponentially small ambiguity of the perturbative
series and the exponentially small weight the prior assigns beyond the pole are
the same quantity, which is a useful consistency check on the whole construction.

\subsection*{7.3 The strongly dissipative limit}

\begin{proposition}
\label{prop:large}
As $\alpha\to\infty$,
\begin{equation}
\langle\eta\rangle(\alpha)\;=\;\frac{\ln\alpha-\gamma}{\alpha}
+O\!\left(\frac{\ln\alpha}{\alpha^2}\right),
\label{eq:large}
\end{equation}
with $\gamma$ the Euler--Mascheroni constant, and therefore
\begin{equation}
\frac{\langle\eta\rangle}{\eta_{\mathrm{det}}}\;\longrightarrow\;\ln\alpha-\gamma
\qquad(\alpha\to\infty).
\label{eq:ratio}
\end{equation}
\end{proposition}

\begin{proof}
Use the convergent expansion
$E_1(z)=-\gamma-\ln z+\sum_{n\ge1}(-1)^{n+1}z^n/(n\cdot n!)$ as $z\to0$ with
$z=1/\alpha$, giving $E_1(1/\alpha)=\ln\alpha-\gamma+\alpha^{-1}+O(\alpha^{-2})$,
together with $e^{1/\alpha}=1+\alpha^{-1}+O(\alpha^{-2})$. Equation~\eqref{eq:ratio}
follows since $\eta_{\mathrm{det}}=\alpha^{-1}+O(\alpha^{-2})$.
\end{proof}

\noindent
Equation~\eqref{eq:ratio} illustrates, in a case where everything is explicit,
how much a fluctuating dissipation can matter. A deterministic machine
dissipating $T_0\langle S\rangle\gg W$ has efficiency falling as $1/\alpha$;
under the maximum-entropy prior a machine with the \emph{same mean dissipation}
has efficiency falling only as $\ln\alpha/\alpha$, an enhancement that grows,
slowly, without limit. The mechanism is convexity: efficiency saturates at $1$,
so the rare realisations in which the machine dissipates far less than average
contribute disproportionately to the mean, and the more strongly dissipative the
machine is on average, the larger that disproportion becomes. The asymptotic form
is accurate to the digits available: at $\alpha=10^8$ the computed ratio is
$17.8435$ against $\ln\alpha-\gamma=17.8435$.

It should be stressed that the \emph{magnitude} $\ln\alpha-\gamma$ is a property
of this particular prior and not a general law. Only the direction of the effect
survives arbitrary priors, by Proposition~\ref{prop:jensen}; how large it may be
is settled not here but in Section~4, where the second moment of the dissipation
bounds it from above without any distributional commitment.

\begin{remark}
Read physically, characterising a strongly dissipative small machine by its mean
dissipation alone systematically \emph{understates} its mean efficiency. That
this happens at all is distribution-free, by Proposition~\ref{prop:jensen}; the
size of the shortfall is not, and depends on how variable the dissipation is ---
as Section~10.1 makes quantitative and Section~4 bounds. The logarithmic growth in
\eqref{eq:ratio} should therefore be read as what the maximum-entropy prior
predicts, and as an indication of the scale the effect can reach, rather than as
a law obeyed by every machine with a given mean dissipation.
\end{remark}

\subsection*{7.4 Summary of the behaviour}

Fig.~\ref{fig:1} collects the results. The mean efficiency interpolates
monotonically between the two limits, lying above the deterministic curve
\eqref{eq:det} everywhere. The absolute gap
$\langle\eta\rangle-\eta_{\mathrm{det}}$ closes as $\alpha^{2}$ at small
$\alpha$, rises to a maximum of about $0.136$ near $\alpha\approx3.4$, and then
decays to zero: at $\alpha=10^{3}$ it is $5.3\times10^{-3}$ and at $\alpha=10^{6}$
it is $1.2\times10^{-5}$. What grows logarithmically at large $\alpha$ is the
\emph{ratio} \eqref{eq:ratio}, since the deterministic value falls faster than
the mean; the two statements are easily confused and only the second is a
statement about how visible the effect is relative to the baseline. Both asymptotic forms are shown
only over the ranges in which they are accurate; the crossover region
$\alpha\sim1$ is described by neither and requires \eqref{eq:main} itself.

\begin{figure}[H]
\centering
\includegraphics{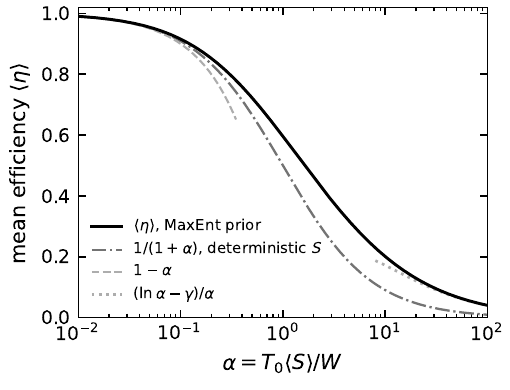}
\caption{Mean efficiency \eqref{eq:main} under the maximum-entropy prior (solid)
against $\alpha=T_0\langle S\rangle/W$, compared with the deterministic value
$1/(1+\alpha)$ (dash-dotted). Light dashed and dotted curves show the
near-reversible form $1-\alpha$ and the strongly dissipative form
$(\ln\alpha-\gamma)/\alpha$, each plotted only where it is accurate. The vertical
separation between the solid and dash-dotted curves is the fluctuation
enhancement of Proposition~\ref{prop:jensen}; on this linear ordinate it peaks
near $\alpha\approx3.4$ and closes at both ends, while the \emph{ratio} of the
two curves grows without bound}
\label{fig:1}
\end{figure}

\section*{8. Relation to existing results}

\subsection*{8.1 What is known about efficiency fluctuations}

The efficiency of a small machine has been studied as a fluctuating quantity for
more than a decade, and the results are considerably sharper than anything
obtained here. Verley \emph{et al.} established that the efficiency of a
stochastic machine obeys a large-deviation principle at long times and that its
rate function has universal features, including the result that the reversible
(Carnot) efficiency is the \emph{least} likely value under symmetric
driving~\cite{Verley2014,VerleyPRE2014}. Manikandan \emph{et al.} extended the
universality analysis from finite state spaces to arbitrary ones and identified
the conditions under which the universal features
fail~\cite{Manikandan2019}. Proesmans and Van den Broeck worked out the
efficiency statistics explicitly in five solvable models~\cite{Proesmans2015},
and Muratore-Ginanneschi and Schwieger obtained the efficiency at maximum power
for sub-micron engines by solving an optimal-transport problem for the driving
protocol~\cite{Muratore2015}. Experimentally, the Brownian Carnot engine of
Mart\'inez \emph{et al.} realises this regime directly~\cite{Martinez2016}.

What these results share is their starting point. In each case a dynamical model
is specified --- a potential and a protocol, a set of transition rates, a
Langevin equation --- and the statistics of efficiency are then \emph{derived}
from it. The output is correspondingly strong: a full distribution, or a rate
function, or an exact optimum.

\subsection*{8.2 What is different here}

The present construction differs in four respects, and it is weaker in all
four.

First, in the works above both the delivered work and the absorbed heat
fluctuate, and the efficiency is the ratio of two random variables. Here, by
Assumption~1, the delivered work is conditioned on and only the dissipation
fluctuates. This is a genuine restriction, not a simplification of the same
problem.

Second, and more importantly, no dynamics are used. The distribution of entropy
production is not derived; it is \emph{assigned}, as the maximum-entropy
distribution consistent with a declared information state. Nothing about the
machine enters beyond the two numbers $W$ and $T_0\langle S\rangle$.

Third, the mathematical machinery is not new: the bounds of Sections~5 and~6 are
classical moment-problem results (Remark~\ref{rem:moment}) evaluated for a
particular $\eta$, whereas the works cited above develop stochastic
thermodynamics proper. Fourth, the output is a single number rather than a
distribution. Section~3
established the price of this: the mean dissipation determines the mean
efficiency only from below, so the number obtained is a benchmark attached to an
information state, not a prediction attached to a machine.

The trade is therefore explicit. The results cited above require a model and
deliver a distribution; the present one requires an energy budget and delivers a
reference value together with a prior-free inequality. These are not competing
answers to one question. They answer different questions, and the second is worth
asking only when the first cannot be.

\subsection*{8.3 Why a mean, when the literature works with distributions}

It is reasonable to ask why the efficiency literature concentrates on
distributions and rate functions rather than on means, and the answer bears
directly on Section~2.

Stochastic efficiency is a ratio whose denominator fluctuates. In the Brownian
engine treated by Proesmans and Van den Broeck the efficiency is, explicitly, the
ratio of two correlated Gaussian variables~\cite{Proesmans2015}; such a ratio has
a Cauchy-like distribution, for which no moment of any order exists. This is
stated as a general property of the conventional figure of merit. The review of
the field records that for $\tilde\eta=-W/Q_{\rm h}$ the moments
$\langle\tilde\eta^{k}\rangle$ diverge for every $k\ge1$, with
$\rho(\tilde\eta)\sim\tilde\eta^{-2}$, because the stochastic input heat in the
denominator can be arbitrarily close to zero~\cite{Holubec2022}; Polettini,
Verley and Esposito put it in the same terms, that the distribution affords no
moments of any order, so that there is neither an average efficiency nor a
mean-square error~\cite{Polettini2015}. Working with $P(\eta)$ and its
large-deviation function is therefore not a stylistic preference but a
necessity.

The present setting escapes it, and the escape is the reason for the choice of
figure of merit. Here the denominator is the consumed exergy $W+T_0S$ rather than
the input heat, so numerator and denominator vanish together and the ratio is
confined to $[0,1]$ whenever $W$ and $S$ are non-negative
(Proposition~\ref{prop:jointexist}); every moment then exists, and the
moment-problem machinery of Sections~4 and~5 has something to act on.
Proposition~\ref{prop:exist} locates the one way back to the pathology --- a
dissipation admitting negative values, which restores a pole inside the physical
region --- and Assumption~\ref{as:pos} excludes it.

This inverts the reading Section~2 might otherwise invite.
Assumption~\ref{as:pos} is not a restriction that narrows an otherwise general
theory; together with the exergetic definition it is the condition under which
a moment-based account of stochastic efficiency is possible, in a setting where
the conventional definition admits none. The paper is not proposing a mean in
place of a distribution as a matter of taste: it is identifying the figure of
merit for which the mean exists, proving exactly when, and then asking what two
measured moments of it determine.

\subsection*{8.4 Relation to the fluctuation theorems}

A clarification is required, because the central object of this paper resembles
one that is already famous. The weight
$e^{-S/\langle S\rangle}$ appearing in the prior \eqref{eq:prior} coincides in
form with the kernel of the fluctuation theorems when $\langle S\rangle=k_B$,
and the integral fluctuation theorem $\langle e^{-S/k_B}\rangle=1$ is among the
best-known results in the field~\cite{Seifert2012}. The resemblance is formal.
Here the weight is a maximum-entropy prior --- an expression of what is not
known about $S$ --- rather than a ratio of forward and reverse path measures. No
fluctuation-theorem result is claimed, and none is used in Sections~3 and~4.

The fluctuation theorem does appear in the paper, but in a different role. In
Appendix~A it is imposed as a \emph{constraint} on the prior, in order
to test whether restricting the support to $S\ge0$ drives the result; the outcome
(Table~\ref{tab:ftprior}) is that it does not. That is a robustness check on an assumption, not a
derivation of a physical law, and the distinction should be kept clear in any
reading of the present work.

There is also a structural reason, independent of the maximum-entropy reading,
why no fluctuation-theorem bound on $\eta$ is available to be used here.
Efficiency is a ratio of two quantities that are both odd under path reversal,
so it is even, $\eta(\Gamma^{\dagger})=\eta(\Gamma)$. Uncertainty relations and
fluctuation-theorem bounds are derived for observables antisymmetric under time
reversal, and the relaxations available do not reach the symmetric case: Indo and
Hasegawa admit observables satisfying
$\phi(\Gamma^{\dagger})\phi(\Gamma)\le0$, which a positive even observable fails
identically~\cite{Indo2025}. Together with the divergence of the moments of
$-W/Q_{\rm h}$ recorded in Section~8.3, this is why the fluctuation literature
bounds efficiency \emph{distributions} rather than efficiency \emph{means}, and
why the ceiling of Section~6.5 is reached by composing an uncertainty relation
for the work current with the moment problem rather than by applying one to
$\eta$ itself.

\subsection*{8.5 Where the two approaches meet}

The natural point of contact is Proposition~\ref{prop:gamma}. Given a dynamical model, one may
compute the distribution of entropy production and, where it is well
approximated within the Gamma family \eqref{eq:gamma}, read off the shape
parameter $k$; \eqref{eq:gammares} then returns the mean efficiency directly.
Conversely, an efficiency measurement together with the mean dissipation
determines an effective $k$ and hence $\langle\ln S\rangle$, a second moment of
the dissipation distribution that is not otherwise easy to access. In that sense
the present construction is best regarded as a coarse complement to the
large-deviation machinery: it extracts limited distributional information from
measurements too crude to support a rate function, and it does so with a stated
inequality --- Proposition~\ref{prop:jensen} --- that holds whatever the underlying dynamics turn
out to be.
\subsection*{8.6 Efficiency bounds from the uncertainty relations}

The uncertainty relations have themselves been used to bound machine efficiency,
and the relationship to Section~6.5 should be stated precisely, because the
bounded object is different in each case.

Pietzonka, Barato and Seifert obtain a universal bound on the efficiency of a
molecular motor~\cite{Pietzonka2016},
$\eta\le\left(1+vk_BT/Df\right)^{-1}$, with $v$ the mean velocity, $D$ the
diffusion coefficient of the displacement and $f$ the external force. The
structural resemblance to the lower bound \eqref{eq:jensen} is real, and so is the
direction of the dependence on dispersion: $D$ sits in the denominator of the
correction, so a noisier output raises their ceiling, just as a noisier
dissipation raises $\mathcal{U}$ here. The statements differ in what fluctuates
and in what is bounded. Theirs constrains the ratio of mean output to mean input
for a motor in a steady state --- a deterministic number --- using the dispersion
of the \emph{output} current. Monnai~\cite{Monnai2025} likewise derives two-sided
bounds on $\eta=\langle W\rangle/\langle Q_1\rangle$, again a ratio of means, and
does not analyse the separation between them.

The object bounded here is $\langle W/(W+T_0S)\rangle$, the mean of a fluctuating
ratio, which differs from the ratio of means by exactly the quantity Section~6.4
studies and which the counterexample there shows can fail in either direction.

The relation to~\cite{Pietzonka2016} is exact rather than analogical.
Proposition~\ref{prop:pareto} composes their dynamical input --- the uncertainty
relation applied to the work current --- with the moment ceiling, and returns
\eqref{eq:pareto}, whose $\varsigma^{2}\to0$ limit \eqref{eq:pbslimit} is their
bound with the correct constant. Their result is therefore the
zero-dissipation-variance member of a one-parameter family that is strictly
increasing in that parameter. The practical consequence is a warning rather than a
generalisation for its own sake: a reader who measures $v$, $D$ and $f$, forms
$\left(1+vk_BT_0/Df\right)^{-1}$ and reads it as a ceiling on the mean trajectory
efficiency is making exactly the substitution Section~6.4 shows to be unsafe, and
\eqref{eq:pareto} says what the correct ceiling is.

The same substitution does a second job that neither reference makes:
Corollary~\ref{cor:widthfloor} floors the \emph{width} of the sharp interval, so
one dynamical inequality both caps the efficiency and bounds how well two
measurements can determine it. Pietzonka et al.\ and Monnai constrain where the
efficiency sits; Proposition~\ref{prop:pareto} corrects that constraint for the
mean of the ratio, and Corollary~\ref{cor:widthfloor} says it cannot be sharpened
away.

One further contrast is worth drawing. Their bounds tighten as the machine becomes
more precise, and are therefore most informative for a well-behaved device.
Corollary~\ref{cor:widthfloor} runs the other way: it is a statement about the
observer's residual ignorance, and it is largest for the most reproducible
machine. Within the class of dynamics the relation admits --- and that class is
narrow, as Section~10.4 sets out --- both statements reach the machine through two
measured relative variances only.

\section*{9. A worked machine: a motor with futile cycles}

The bounds of Sections~3--6 are statements about an information state, and the
question they invite is whether that information state ever binds: for an actual
machine, with an actual dissipation distribution, is \eqref{eq:interval}
informative, and is its upper edge reachable? This section answers both for a
standard chemomechanical motor. The model is exactly solvable, so nothing below
is simulated.

\subsection*{9.1 The model and its exact moments}

A motor consumes free energy $\Delta\mu$ per turnover. With probability $p$ the
turnover is \emph{productive}: it advances the load through one step, delivering
work $w<\Delta\mu$ and dissipating $\Delta\mu-w$. With probability $1-p$ it is
\emph{futile}: the chemical step completes, no work is delivered, and the whole
of $\Delta\mu$ is dissipated. Futile cycles of this kind are the standard
coarse-grained description of slippage in molecular motors, and $p$ is fixed by
the ratio of the two outgoing rates at the branch point.

The machine is observed until it has delivered $N$ productive steps. This
stopping rule is what puts the model in the information state of Section~2: the
delivered work
\begin{equation}
W\;=\;N\,w
\label{eq:motorW}
\end{equation}
is deterministic, while the dissipation fluctuates through the number $M$ of
futile cycles, which is negative-binomial with $N$ successes and success
probability $p$. Writing $T_0S=N(\Delta\mu-w)+M\,\Delta\mu$ and using
$\langle M\rangle=N(1-p)/p$, $\mathrm{Var}(M)=N(1-p)/p^{2}$,
\begin{equation}
\alpha\;=\;\frac{(\Delta\mu-w)+\Delta\mu\,(1-p)/p}{w},
\qquad
\sigma^{2}\;=\;\frac{1}{N}\,
\frac{\Delta\mu^{2}(1-p)/p^{2}}{\left[(\Delta\mu-w)+\Delta\mu(1-p)/p\right]^{2}} .
\label{eq:motorparams}
\end{equation}
The loss fraction $\alpha$ is set by the chemistry and the load alone; the
relative variance carries the whole of the size dependence and falls as $1/N$,
which is Section~10.1 realised in a specific machine. The mean efficiency
$\langle\eta\rangle=\langle W/(W+T_0S)\rangle$ is a convergent sum over the
negative-binomial law and is evaluated below to machine precision.

\subsection*{9.2 A dissipation floor sharpens the ceiling}

This motor cannot dissipate nothing: each productive step costs at least
$\Delta\mu-w$, so $T_0S\ge N(\Delta\mu-w)$ with certainty. In the normalised
variable $x=S/\langle S\rangle$ of \eqref{eq:etax} this is a known floor
$x\ge b$ with
\begin{equation}
b\;=\;\frac{\Delta\mu-w}{(\Delta\mu-w)+\Delta\mu(1-p)/p}\;\in\;[0,1).
\label{eq:motorb}
\end{equation}
Proposition~\ref{prop:upper} does not use such a floor, and pays for it: its
extremal law \eqref{eq:extremal} places an atom at $x=0$, which this machine
forbids. The floor is itself a measurable number --- the reversible cost of the
delivered output, divided by the measured mean dissipation --- and it tightens
the ceiling.

\begin{proposition}[Upper bound with a known dissipation floor]
\label{prop:upperfloor}
Let $x\ge b$ almost surely, with $0\le b<1$, $\langle x\rangle=1$ and
$\langle x^{2}\rangle=1+\sigma^{2}$, and let $\eta(x)=(1+\alpha x)^{-1}$ with
$\alpha>0$. Put
\begin{equation}
q\;=\;\frac{\sigma^{2}}{\sigma^{2}+(1-b)^{2}},
\qquad
c\;=\;1+\frac{\sigma^{2}}{1-b}\;>\;1 .
\label{eq:floorqc}
\end{equation}
Then
\begin{equation}
\langle\eta\rangle\;\le\;\mathcal{U}_{b}(\alpha,\sigma^{2})
\;\equiv\;\frac{q}{1+\alpha b}+\frac{1-q}{1+\alpha c},
\label{eq:upperfloor}
\end{equation}
and the bound is attained, by the two-point law with mass $q$ at $b$ and $1-q$
at $c$. At $b=0$ it reduces to Proposition~\ref{prop:upper}.
\end{proposition}

\begin{proof}
Let $Q$ be the quadratic determined by $Q(b)=\eta(b)$, $Q(c)=\eta(c)$ and
$Q'(c)=\eta'(c)$. Since $1+\alpha x>0$ on the support, the sign of $Q-\eta$ is
that of $Q(x)(1+\alpha x)-1$, and expanding the interpolation conditions gives
the identity
\begin{equation}
Q(x)\left(1+\alpha x\right)-1
\;=\;\frac{\alpha^{3}}{(1+\alpha b)(1+\alpha c)^{2}}\;(x-b)\,(x-c)^{2},
\label{eq:floorcubic}
\end{equation}
which is non-negative for every $x\ge b$. Hence $Q\ge\eta$ pointwise on the
support and $\langle\eta\rangle\le\langle Q\rangle$, a quantity depending on the
distribution only through $\langle x\rangle$ and $\langle x^{2}\rangle$. The
two-point law $\Pr[x=b]=q$, $\Pr[x=c]=1-q$ with \eqref{eq:floorqc} has exactly
those two moments, so $\langle Q\rangle$ equals the right-hand side of
\eqref{eq:upperfloor}; equality in $Q\ge\eta$ requires by \eqref{eq:floorcubic}
that the support lie in $\{b,c\}$, and that law attains the bound. Setting
$b=0$ gives $q=\sigma^{2}/(1+\sigma^{2})$, $c=1+\sigma^{2}$ and recovers
\eqref{eq:upper}.
\end{proof}

\noindent
The bound was checked numerically in two independent ways: random discrete laws
satisfying the three constraints never exceed \eqref{eq:upperfloor}, and direct
maximisation over four-point laws reaches it to six decimal places at every
$(\alpha,\sigma^{2},b)$ tested. The mechanism is the one Section~4.3 identifies:
the ceiling is set by how much probability can sit at the smallest admissible
dissipation, and a floor moves that point away from zero.

\subsection*{9.3 What the bounds say for this motor}

Table~\ref{tab:motor} evaluates the interval for $\Delta\mu=1\,k_{\mathrm B}T$
and $w=0.6\,k_{\mathrm B}T$, so that a productive step costs at least
$0.4\,k_{\mathrm B}T$. The exact mean efficiency lies strictly inside
\eqref{eq:interval} at every operating point, and the floor
\eqref{eq:upperfloor} removes between $40$ and $67$ per cent of the width. Two
moments and one thermodynamic constant therefore locate $\langle\eta\rangle$ of
this machine to a few per cent at $N=5$ and to better than a per cent at
$N=100$, with no prior and no solution of the dynamics.

\begin{table}[H]
\centering
\caption{The motor of Section~9.1 at $\Delta\mu=1$, $w=0.6$. Floor is
Proposition~\ref{prop:jensen}, $\mathcal U$ is Proposition~\ref{prop:upper},
$\mathcal U_b$ is Proposition~\ref{prop:upperfloor}; the last column is the
fraction of the interval width that the dissipation floor removes. The exact
value is a sum over the negative-binomial law, not a simulation}
\label{tab:motor}
\begin{tabular}{rrrrrrrrr}
\toprule
$N$ & $p$ & $\alpha$ & $\sigma^{2}$ & $b$ & floor & $\langle\eta\rangle$ &
$\mathcal U$ & $\mathcal U_{b}$ \\
\midrule
1 & 0.50 & 2.3333 & 1.0204 & 0.2857 & 0.30000 & 0.41589 & 0.59167 & 0.45000 \\
1 & 0.80 & 1.0833 & 0.7396 & 0.6154 & 0.48000 & 0.53554 & 0.62444 & 0.54000 \\
1 & 0.95 & 0.7544 & 0.2704 & 0.8837 & 0.57000 & 0.58474 & 0.61479 & 0.58500 \\
2 & 0.80 & 1.0833 & 0.3698 & 0.6154 & 0.48000 & 0.51564 & 0.56387 & 0.52000 \\
5 & 0.80 & 1.0833 & 0.1479 & 0.6154 & 0.48000 & 0.49699 & 0.51714 & 0.50000 \\
5 & 0.95 & 0.7544 & 0.0541 & 0.8837 & 0.57000 & 0.57482 & 0.57977 & 0.57500 \\
20 & 0.80 & 1.0833 & 0.0370 & 0.6154 & 0.48000 & 0.48466 & 0.48981 & 0.48571 \\
100 & 0.80 & 1.0833 & 0.0074 & 0.6154 & 0.48000 & 0.48095 & 0.48199 & 0.48119 \\
\bottomrule
\end{tabular}
\end{table}

\subsection*{9.4 The extremal machine is a physical one}

Proposition~\ref{prop:upperfloor} is attained by a two-point dissipation law,
and this motor produces one. At $N=1$ the number of futile cycles is
$M=0,1,2,\dots$ with geometric weights; as $p\to1$ the weight beyond $M=1$
becomes negligible and the dissipation becomes two-valued --- the extremal law
of the proposition, with the atom at the floor rather than at zero. A motor that
almost never slips is therefore not merely close to the bound but converges to
it, as Table~\ref{tab:attain} shows: the ratio
$\langle\eta\rangle/\mathcal U_{b}$ is $0.9917$ at $p=0.8$ and $1.000000$ to six
decimals at $p=0.999$.

This is the converse of the reading that Section~4.3 gives of
\eqref{eq:extremal}. There the extremal machine was intermittently reversible
and could be dismissed as an idealisation; here the same extremal structure
appears, with the reversible atom replaced by the thermodynamic floor, in a
machine defined by rates. Saturation of the ceiling is thus a statement about
operating regime --- rare, all-or-nothing losses --- rather than about an
unattainable limit.

\begin{table}[H]
\centering
\caption{Approach to the ceiling at $N=1$ as futile cycles become rare. The
dissipation law tends to the two-point extremal law of
Proposition~\ref{prop:upperfloor}, and the bound is saturated}
\label{tab:attain}
\begin{tabular}{rrrrrr}
\toprule
$p$ & $\alpha$ & $\sigma^{2}$ & $b$ & $\langle\eta\rangle$ &
$\langle\eta\rangle/\mathcal U_{b}$ \\
\midrule
0.500 & 2.3333 & 1.02041 & 0.28571 & 0.415888 & 0.924196 \\
0.800 & 1.0833 & 0.73964 & 0.61538 & 0.535545 & 0.991749 \\
0.900 & 0.8519 & 0.47259 & 0.78261 & 0.568947 & 0.998152 \\
0.950 & 0.7544 & 0.27042 & 0.88372 & 0.584744 & 0.999562 \\
0.990 & 0.6835 & 0.06067 & 0.97537 & 0.596990 & 0.999983 \\
0.999 & 0.6683 & 0.00623 & 0.99750 & 0.599700 & 1.000000 \\
\bottomrule
\end{tabular}
\end{table}

\subsection*{9.5 The size of the correction, checked against a machine}

Proposition~\ref{prop:size} predicts that the fluctuation correction to the
deterministic accounting is $\alpha^{2}\sigma^{2}/(1+\alpha)^{3}$ to leading
order. The motor provides an independent test, since $\sigma^{2}=O(1/N)$ is
generated by the dynamics rather than assumed. At $p=0.8$ the measured excess
$\langle\eta\rangle-(1+\alpha)^{-1}$ divided by that prediction is $0.970$ at
$N=20$, $0.994$ at $N=100$ and $0.9985$ at $N=400$; at $p=0.95$ the
corresponding ratios are $0.957$, $0.991$ and $0.998$. The correction is a few
tenths of a per cent of the efficiency at $N=100$, and of order five per cent at
$N=1$: the regime in which the construction says something a deterministic
accounting does not is the regime of a handful of turnovers, exactly as
Section~10.1 argues.

\section*{10. Discussion}

\subsection*{10.1 When does the effect matter?}

The dimensionless group $\alpha=T_0\langle S\rangle/W$ is not by itself a measure
of smallness: a macroscopic engine losing a tenth of its input has
$\alpha\approx0.1$, comfortably within the range plotted in Fig.~\ref{fig:2}. What
distinguishes a small machine is not the size of $\alpha$ but the size of the
\emph{relative fluctuations} of the dissipation, and it is those that the
corrections of Sections~3 and~4 track.

\begin{proposition}[Magnitude of the fluctuation correction]
\label{prop:size}
For a prior with unit-mean shape parameter $k$, i.e.\ with relative variance
$\mathrm{Var}(S)/\langle S\rangle^2 = 1/k$,
\begin{equation}
\langle\eta\rangle_k-\eta_{\mathrm{det}}
\;=\;\frac{\alpha^2}{k\,(1+\alpha)^3}+O(k^{-2}),
\qquad
\frac{\langle\eta\rangle_k-\eta_{\mathrm{det}}}{\eta_{\mathrm{det}}}
\;=\;\frac{\alpha^2}{k\,(1+\alpha)^2}+O(k^{-2}).
\label{eq:size}
\end{equation}
\end{proposition}

\begin{proof}
Expand $\langle\eta(x)\rangle$ about $x=1$: since $\langle x\rangle=1$ the linear
term vanishes and $\langle\eta\rangle=\eta(1)+\tfrac12\eta''(1)\mathrm{Var}(x)+\cdots$
with $\eta''(1)=2\alpha^2(1+\alpha)^{-3}$ and $\mathrm{Var}(x)=1/k$.
\end{proof}

\noindent
Numerically the ratio of \eqref{eq:size} to the exact value from \eqref{eq:gammares} is
$0.987$ at $k=20$ and $0.9975$ at $k=100$, for $\alpha=1$. The practical content
is immediate. If a machine's dissipation is the sum of $N$ roughly independent
contributions then $k\sim N$, and the fluctuation correction falls as $1/N$: for
a macroscopic device it is unmeasurable, and $\eta_{\mathrm{det}}$ is the whole
story. The correction reaches tens of per cent only when $k$ is of order unity,
that is when the dissipation varies by of order $100\%$ from realisation to
realisation. That is the regime of single molecular machines and sub-micron
engines, and it is the only regime in which the present construction says
anything a deterministic accounting would not.

\subsection*{10.2 What could be measured}

Three uses follow, of quite different strength.

The first is a falsification test that involves no prior at all. By
Proposition~\ref{prop:jensen}, any machine whose measured mean efficiency satisfies
$\langle\eta\rangle<(1+\alpha)^{-1}$, with $\alpha$ obtained from the measured
work and mean dissipation, contradicts the framework outright --- either the
energy accounting \eqref{eq:budget} is wrong for that device, or $S$ is not
non-negative in the required coarse-grained sense, or the two measurements do not
refer to the same process. Because Proposition~\ref{prop:jensen} assumes nothing about the
distribution, this test cannot be evaded by adjusting a prior.

The second is inferential rather than predictive. Given measurements of
$\langle\eta\rangle$ and $\alpha$, equation~\eqref{eq:gammares} can be inverted for the
effective shape parameter $k$, which by the construction of Section~3.2 is in
one-to-one correspondence with $\langle\ln S\rangle$. An efficiency measurement
therefore reports a second functional of the dissipation distribution beyond its
mean. This is a modest quantity, but it is accessible from data too coarse to
support the reconstruction of a rate function, which is the situation in which
this framework is intended to be used at all.

The third is the sharpest, and it is a prediction about the trajectory record
rather than a statement about an efficiency. By
Proposition~\ref{prop:stability} the shortfall
$\mathcal{U}(\alpha,\sigma^{2})-\langle\eta\rangle$ is exactly $\langle\psi\rangle$,
and $\psi$ vanishes only at $x=0$ and $x=1+\sigma^{2}$. A machine measured close
to the ceiling must therefore spend a controlled fraction of its realisations
near zero dissipation and the remainder near $(1+\sigma^{2})\langle S\rangle$:
the dissipation record must be bimodal. Equation~\eqref{eq:stability} makes this
quantitative without requiring the bound to be attained --- at
$\alpha=\sigma^{2}=1$, a machine within $10^{-3}$ of the ceiling has at most
$6\%$ of its dissipation further than $0.5$ from $\{0,2\langle S\rangle\}$. The
test uses no prior, no dynamical model, and no measurement beyond the histogram
from which $\alpha$ and $\sigma^{2}$ were formed in the first place.

\subsubsection*{A concrete assay}

The framework needs the first two moments of $S$, or of $R=T_0S/W$, over repeated
observation windows, and single-molecule experiments already report the
ingredients. For a rotary motor such as $F_1$-ATPase the accounting closes
without a kinetic model: the chemical potential difference $\Delta\mu$ per
hydrolysis is fixed by the nucleotide concentrations in the buffer, the number
$n$ of hydrolysis events in a window is counted, and the delivered work
$W=\int\tau\,\mathrm{d}\theta$ is read from the trapped bead, so that the
budget \eqref{eq:budget} gives
\begin{equation}
T_0S\;=\;n\,\Delta\mu\;-\;W
\label{eq:assay}
\end{equation}
realisation by realisation. Toyabe \emph{et al.} measure these three quantities
together for $F_1$-ATPase~\cite{Toyabe2011}. Repeating over windows yields the
histogram of $S$, hence $\alpha$, $\sigma^{2}$, the bracket
\eqref{eq:interval}, and --- from the same histogram --- the bimodality test
above. For a translational motor such as kinesin the accounting is identical with
$W=F\,\Delta x$ against the load of an optical force clamp.

Two remarks attach to this. First, Assumption~\ref{as:pos} holds here in the
structural sense of Section~2.1(i) rather than by appeal to aggregation: the
motor consumes its fuel in the forward direction, so $n\Delta\mu\ge W$
realisation by realisation, whatever the size of the fluctuations. Second,
\eqref{eq:assay} determines the capture fraction of
Assumption~\ref{as:wcurrent}. It is $r=1$ whenever the counted fuel consumption
is complete, because the budget then charges the device with every unit of exergy
it received, including dissipation no probe resolves. What defeats this is not
hidden dissipation but hidden \emph{consumption}: futile or unresolved hydrolyses
raise the true input above $n\Delta\mu$ and leave $r<1$. That the energy balance
of these motors is in practice hard to close is documented --- for kinesin, Ariga,
Tomishige and Mizuno find the measured heat and work together fall well short of
the input free-energy change, and attribute the remainder to internal
dissipation~\cite{Ariga2018} --- so $r$ should be argued for in a given assay
rather than assumed.

\subsection*{10.3 Quantum machines}

Nothing in Sections~3--6 refers to classical dynamics. The bounds are statements
about a non-negative random variable and a convex function of it, so they apply
to a quantum device on exactly two conditions: that the accounting
\eqref{eq:budget} holds with a single dead-state temperature, and that the
entropy production is a genuine random variable satisfying
Assumption~\ref{as:pos}. It is worth being precise about when the second
condition holds, because that is where quantum mechanics makes trouble.

Under the two-point measurement scheme --- projective energy measurements before
and after the process --- work and entropy production are bona fide classical
random variables, the quantum fluctuation relations take their familiar
form~\cite{CampisiHanggiTalkner}, and every result of this paper transfers
without modification. The same is true of any protocol whose statistics are
obtained by measurement in the energy eigenbasis.

With initial coherence in that basis the situation changes, and not merely
technically. A no-go theorem forbids a genuine probability distribution for
quantum work that simultaneously respects the first law and reduces to the
two-point measurement result; the standard resolution is to admit a
\emph{quasi}probability, which may take negative or complex
values~\cite{PeiChenQuan}. Where that is the case the present results do not
merely go unproven --- they are not well posed, because the moment problem of
Sections~4 and~5 requires a probability measure to act on, and no such measure
exists. Extending the bracket \eqref{eq:interval} to coherent quantum machines
would therefore require a moment theory for signed measures, which is a different
undertaking from the one carried out here.

Within the measurement-based setting, however, the transfer is immediate, and
the parameters retain their meaning: $\alpha$ is the mean dissipated exergy per
unit useful work and $\sigma^{2}$ the relative variance of the measured entropy
production.

\subsection*{10.4 Scope of the dynamical bound}

Proposition~\ref{prop:pareto} and Corollary~\ref{cor:widthfloor} are the only
results here that import a hypothesis about the dynamics, and they inherit every
restriction attached to the inequality they use. Because those restrictions decide
when the bound may be quoted, they are set out rather than summarised.

\emph{Two conditions beyond the uncertainty relation.} The delivered work must
fluctuate, so Assumption~\ref{as:task} has to be relaxed --- which is why the
result appears in Section~6 and not in Section~4 --- and the covariance condition
$\mathrm{Cov}(R,W)\le0$ of Lemma~\ref{lem:cov} must hold. The second is checkable
from the same data the bound uses, and it fails when the dissipation grows
superlinearly in the delivered work; it should be verified, not assumed. The
capture fraction $r$ is not directly observable, and the bound weakens
monotonically as $r$ falls, so quoting it at $r=1$ presumes that the exergy
accounting charges the device with all of the entropy produced.
Section~10.2 gives the
condition under which that presumption is safe in a single-molecule assay, and
the failure mode that breaks it.

\emph{Network topology is not a restriction, but it is coupled to $r$.} The
uncertainty relation holds for any Markov jump network in a nonequilibrium steady
state, so nothing in Proposition~\ref{prop:pareto} confines the machine to a
single working cycle: competing and futile cycles sharing a state space are
admitted, and the bound is insensitive to the correlations between them. The two
conditions interact, however, and not favourably. A network with slip pathways is
also the setting in which the counted work current is charged with less than the
total entropy production, so precisely those machines whose topology the relation
covers most generally are the ones for which $r<1$ has to be taken seriously.

\emph{Steady state, and steady-state initialisation.} The finite-time
relation~\cite{Horowitz2017} requires stationarity throughout the observation
window, initialisation included; a machine watched from a prepared state over a
single cycle lies outside it. A periodically driven machine is admitted only for a
time-symmetric protocol. Under general time-dependent driving the relation fails
--- Koyuk and Seifert exhibit the failure in the slow-driving limit, where a
current's dispersion becomes negligible while heat continues to be
dissipated~\cite{Koyuk2020} --- and the generalised relation that replaces it
carries a response factor that can vanish, so a driven engine may admit no bound
from this argument at all.

\emph{Known failures.} The uncertainty relation is not universal. It fails for
Markov processes with unidirectional transitions~\cite{Pal2021}, for ballistic and
coherent transport, where time-reversal-breaking fields weaken it
sharply~\cite{Brandner2018}, and --- in a mechanically transparent counterexample
--- for classical pendulum clocks, whose precision is limited by dissipation far
more weakly than the relation permits~\cite{Pietzonka2022}; for time-discrete
Markov chains it holds in modified form~\cite{Proesmans2017}. Each exception is
inherited here.

\emph{What a fluctuation theorem alone gives.} Under the much weaker hypothesis
that some fluctuation theorem holds, a lower bound on $\rho$ survives but a far
weaker one. Timpanaro, Guarnieri, Goold and Landi give the tightest bound
derivable from the exchange fluctuation theorems~\cite{Timpanaro2019},
$\epsilon_{w}^{2}\ge\mathrm{csch}^{2}\!\left(g(\langle\Sigma_{\rm tot}\rangle/2)\right)$
with $g$ the inverse of $x\tanh x$, valid for classical and quantum systems
undergoing non-Markovian and non-stationary processes; Hasegawa and Vu give the
simpler $\epsilon_{w}^{2}\ge2/(e^{\langle\Sigma_{\rm tot}\rangle}-1)$~\cite{Hasegawa2019}.
Both agree with the steady-state relation as $\langle\Sigma_{\rm tot}\rangle\to0$
and decay exponentially thereafter --- a factor $1.5$ weaker at
$\langle\Sigma_{\rm tot}\rangle=1$, four times at $3$, eleven hundred times at
$10$. Inverted, they bound $\langle\Sigma_{\rm tot}\rangle$, and hence $\rho$,
below by a quantity growing only logarithmically in $1/\epsilon_{w}^{2}$: outside
the steady-state hypotheses the ceiling survives in form but not in strength.

\emph{The bound is not sharp.} Equality in \eqref{eq:pareto} requires the moment
ceiling to be attained, \emph{and} $\rho$ to sit at the uncertainty-relation
infimum, \emph{and} $\mathrm{Cov}(R,W)=0$. The last forces $R$ to be constant and
hence $\varsigma^{2}=0$, which is exactly where the first two can hold together;
away from that corner \eqref{eq:pareto} is a valid ceiling but not an attained
one. Whether the attainable maximum --- the moment problem restricted to laws a
dynamics obeying the relation can produce --- differs materially from it is open,
and is the natural next question.

\subsection*{10.5 Limitations}

Four restrictions should be kept in view, three of them stated as assumptions in
Section~2.

\emph{Access to the moments.} Section~6 shows that the bounds themselves do not
require Assumption~\ref{as:task}: they hold for a jointly fluctuating $(W,S)$ on
replacing the moments of the dissipation by those of the ratio $R=T_0S/W$. The
limitation is one of measurement rather than validity. Those moments must be
formed realisation by realisation --- the ratio of separately measured averages
is not a valid substitute, and Section~6.4 exhibits joint laws for which
substituting it yields a floor twice the true mean efficiency --- and
$\langle R\rangle$ may diverge even when $\langle W\rangle$ and $\langle S\rangle$
are both finite, in which case the bounds are true but empty.

\emph{Support of the dissipation.} Assumption~2 restricts $S$ to non-negative
values, and Proposition~\ref{prop:exist} shows this is not optional: without it the mean does
not exist. The quantity computed is therefore a property of coarse-grained
dissipation, over a cycle or an ensemble, and not of individual trajectories,
for which negative excursions are precisely the interesting feature.

\emph{What the accounting requires.} Equation~\eqref{eq:budget} needs a single
reference (dead-state) temperature $T_0$ and the \emph{total} entropy generated;
it does not require the device to touch only one reservoir. The Gouy--Stodola
relation reads $W_{\mathrm{lost}}=T_0S_{\mathrm{gen}}$ whatever the internal
reservoir structure, so \eqref{eq:budget} is an exergy balance and covers
multi-reservoir devices provided the input is measured as exergy rather than as
raw heat. What is genuinely required is that the \emph{task} be the fixed
quantity, in the sense of Assumption~\ref{as:task}; see the remark below.

\emph{The value is not determined.} Corollary~1 remains the governing
limitation: mean dissipation bounds mean efficiency from below and no more. Every
numerical value in this paper is attached to a declared information state.

\subsection*{10.6 Outlook}

Two extensions suggest themselves, one of which is well defined enough to state
precisely.

The moment hierarchy of Sections~5 and~6 terminates in a specific sense: by
\eqref{eq:m3min} the bracket closes at the smallest admissible third moment, so
no fourth moment is needed to pin the value there. What remains open is the
behaviour in between --- whether a fourth moment tightens \eqref{eq:bracket3}
usefully at intermediate $m_{3}$, and whether the extremal laws remain two-point
as the hierarchy is continued.

The substantive extension is to the case in which $\langle R\rangle$ diverges.
Section~6 shows that the bounds hold for a jointly fluctuating $(W,S)$ but become
vacuous when the density of $W$ does not vanish at the origin --- exactly the
regime in which a machine occasionally delivers almost no work. Since $\eta$
remains bounded there by Proposition~\ref{prop:jointexist}, the mean efficiency
is well defined and merely inaccessible to the present method; bounding it would
require functionals of the joint law other than moments of the ratio, such as a
truncated moment or a bound on $\Pr[W<w]$ near the origin. That is the natural
next problem, and it connects directly to the ratio-of-random-variables structure
of the efficiency-fluctuation literature~\cite{Proesmans2015}; unlike the moment
hierarchy, it is not a matter of continuing an existing construction.

\appendix
\renewcommand{\theequation}{A.\arabic{equation}}
\setcounter{equation}{0}
\renewcommand{\theproposition}{A.\arabic{proposition}}
\setcounter{proposition}{0}

\section*{Appendix A. The fluctuation-theorem prior}
\addcontentsline{toc}{section}{Appendix A}

Section~2 restricts the prior to $S\ge0$, and Proposition~\ref{prop:exist} shows
that some such restriction is unavoidable if a mean efficiency is to exist. This
appendix addresses the separate question of whether the \emph{particular} prior
used on that half-line drives the results, by constructing the maximum-entropy
prior that admits negative entropy production while respecting the integral
fluctuation theorem, and comparing the two.

Throughout we work in the dimensionless entropy production $u = S/k_B$, in which
the fluctuation theorem takes its natural form, with density $q(u)$. Since
$H_S[p]=H_u[q]+\ln k_B$, the constant offset does not affect the maximiser.

\subsection*{A.1 The variational problem}

We maximise
\begin{equation}
H[q]\;=\;-\int_{\mathbb{R}}q(u)\ln q(u)\,\mathrm{d}u
\end{equation}
over densities on all of $\mathbb{R}$ --- the support is \emph{not} restricted
--- subject to
\begin{equation}
\text{(C0)}\ \int q\,\mathrm{d}u=1,
\qquad
\text{(C1)}\ \int e^{-u}q\,\mathrm{d}u=1,
\qquad
\text{(C2)}\ \int u\,q\,\mathrm{d}u=m,
\end{equation}
where (C1) is the integral fluctuation theorem $\langle e^{-S/k_B}\rangle=1$
\cite{Seifert2012} and $m=\mu/k_B$.

\subsection*{A.2 Stationarity}

With multipliers $\lambda_0,\lambda_1,\lambda_2$, the first variation of
\[
H[q]-\lambda_0\!\left(\textstyle\int q-1\right)
     -\lambda_2\!\left(\textstyle\int e^{-u}q-1\right)
     -\lambda_1\!\left(\textstyle\int uq-m\right)
\]
vanishes when $-\ln q(u)-1-\lambda_0-\lambda_1u-\lambda_2e^{-u}=0$, since
$\delta(-q\ln q)/\delta q=-\ln q-1$. Absorbing the constant into the
normalisation,
\begin{equation}
q(u)\;=\;\frac{1}{Z}\exp\!\left(-\lambda_1u-\lambda_2e^{-u}\right).
\label{eq:Aform}
\end{equation}
Each constraint contributes its own function to the exponent, as expected.

\subsection*{A.3 Normalisation}

\begin{proposition}
\label{prop:AZ}
$\displaystyle Z(\lambda_1,\lambda_2)=\int_{\mathbb{R}}
e^{-\lambda_1u-\lambda_2e^{-u}}\,\mathrm{d}u
=\Gamma(\lambda_1)\,\lambda_2^{-\lambda_1}$, convergent if and only if
$\lambda_1>0$ and $\lambda_2>0$.
\end{proposition}

\begin{proof}
Substitute $v=e^{-u}$, so $u=-\ln v$ and $\mathrm{d}u=-\mathrm{d}v/v$; as $u$
runs over $(-\infty,\infty)$, $v$ runs over $(\infty,0)$. Since
$e^{-\lambda_1u}=v^{\lambda_1}$,
\[
Z=\int_0^\infty v^{\lambda_1}e^{-\lambda_2v}\,\frac{\mathrm{d}v}{v}
 =\int_0^\infty v^{\lambda_1-1}e^{-\lambda_2v}\,\mathrm{d}v
 =\frac{\Gamma(\lambda_1)}{\lambda_2^{\lambda_1}} .
\]
Convergence at $v\to0$ requires $\lambda_1>0$; at $v\to\infty$ it requires
$\lambda_2>0$. In the original variable, the factor $e^{-\lambda_2e^{-u}}$
controls $u\to-\infty$ and $e^{-\lambda_1u}$ controls $u\to+\infty$.
\end{proof}

\noindent
The same substitution identifies the law of $v$:
\begin{equation}
q(u)\,\mathrm{d}u=\frac{\lambda_2^{\lambda_1}}{\Gamma(\lambda_1)}
v^{\lambda_1-1}e^{-\lambda_2v}\,\mathrm{d}v ,
\end{equation}
so $v=e^{-S/k_B}$ is Gamma distributed with shape $\lambda_1$ and rate
$\lambda_2$. This is what makes the remaining steps elementary.

\subsection*{A.4 Fixing the multipliers}

For a Gamma law $\langle v\rangle=\lambda_1/\lambda_2$, so the fluctuation
theorem (C1) forces
\begin{equation}
\lambda_2=\lambda_1\equiv a ,
\end{equation}
leaving a one-parameter family: $v\sim\mathrm{Gamma}(a,a)$, of unit mean. Using
$\mathbb{E}[\ln v]=\psi(a)-\ln b$ for $\mathrm{Gamma}(a,b)$, the mean constraint
(C2) reads
\begin{equation}
m\;=\;\langle u\rangle\;=\;-\mathbb{E}[\ln v]\;=\;\ln a-\psi(a)\;\equiv\;\varphi(a).
\label{eq:Aphi}
\end{equation}

\begin{proposition}
\label{prop:Aphi}
$\varphi(a)=\ln a-\psi(a)$ is strictly decreasing on $(0,\infty)$, with
$\varphi(0^+)=+\infty$ and $\varphi(\infty)=0^+$. Hence \eqref{eq:Aphi} has a
unique solution $a>0$ for every $m>0$, and none for $m\le0$.
\end{proposition}

\begin{proof}
$\varphi'(a)=1/a-\psi'(a)$, and $\psi'(a)=\sum_{n\ge0}(a+n)^{-2}>1/a$, so
$\varphi'<0$. The limits follow from $\psi(a)=-1/a-\gamma+O(a)$ as $a\to0$ and
$\psi(a)=\ln a-1/(2a)+O(a^{-2})$ as $a\to\infty$, the latter giving
$\varphi(a)\simeq1/(2a)$.
\end{proof}

\begin{remark}
The restriction $m>0$ in Proposition~\ref{prop:Aphi} is not an additional
assumption. By convexity of $x\mapsto e^{-x}$ and Jensen's inequality, (C1) gives
$1=\langle e^{-u}\rangle\ge e^{-\langle u\rangle}$ and hence $\langle
u\rangle\ge0$, with equality only for $u$ degenerate at $0$. The second law
emerges from the fluctuation-theorem constraint, and the maximum-entropy family
reproduces exactly the admissible range: although its support is all of
$\mathbb{R}$, this prior cannot represent a violation of the second law on
average.
\end{remark}

\subsection*{A.5 The stationary point is the maximiser}

$H$ is strictly concave and (C0)--(C2) are linear functionals of $q$, so the
stationary point is the unique global maximiser. Explicitly, for bounded $h$ with
$\langle h\rangle_q=\langle e^{-u}h\rangle_q=\langle u\,h\rangle_q=0$, the
perturbation $q_\varepsilon=q(1+\varepsilon h)$ satisfies all three constraints
exactly for every $\varepsilon$, and
\begin{equation}
H[q_\varepsilon]\;=\;H[q]-\frac{\varepsilon^2}{2}\int q\,h^2\,\mathrm{d}u
+O(\varepsilon^3),
\label{eq:Asecond}
\end{equation}
the first-order term vanishing because $\ln q$ is a linear combination of
$1,u,e^{-u}$, to each of which $h$ is orthogonal. The entropy at the maximum is
\begin{equation}
H[q]\;=\;\ln Z+a\langle u\rangle+a\langle e^{-u}\rangle
\;=\;\ln\Gamma(a)-a\ln a+a\left(\varphi(a)+1\right).
\end{equation}

\subsection*{A.6 Numerical confirmation}

Proposition~\ref{prop:AZ} agrees with quadrature to $10^{-15}$ relative for
$(\lambda_1,\lambda_2)$ away from the convergence boundary. With
$\lambda_2=\lambda_1=a$ the constraints are satisfied to machine precision and
$\langle u\rangle$ matches $\ln a-\psi(a)$ to $10^{-12}$ or better
(Table~\ref{tab:Aconstraints}). For $a=1.7$ the closed-form entropy gives
$H=1.2496609153$ against $1.2496609153$ by quadrature. Taking $h$ a bounded
combination of Gaussian bumps orthogonalised against $\{1,e^{-u},u\}$, the drop
$H[q]-H[q_\varepsilon]$ was positive at every $\varepsilon$ tested and matched
\eqref{eq:Asecond} with ratios $0.945$, $0.970$, $0.985$, $0.992$, $0.996$ for
$\varepsilon=0.4,0.2,0.1,0.05,0.025$, converging to unity as $\varepsilon\to0$.

\begin{table}[H]
\centering
\caption{Constraint check for $v=e^{-S/k_B}\sim\mathrm{Gamma}(a,a)$, by
quadrature. Note $\langle u\rangle=\gamma$ at $a=1$: the unit-rate case
corresponds to mean dissipation $\mu=\gamma k_B$}
\label{tab:Aconstraints}
\begin{tabular}{ccccc}
\toprule
$a$ & $\int q\,\mathrm{d}u$ & $\langle e^{-u}\rangle$ & $\langle u\rangle$ & $\ln a-\psi(a)$\\
\midrule
$0.5$ & $1.0000000000$ & $1.0000000000$ & $1.27036285$ & $1.27036285$\\
$1.0$ & $1.0000000000$ & $1.0000000000$ & $0.57721566$ & $0.57721566$\\
$2.0$ & $1.0000000000$ & $1.0000000000$ & $0.27036285$ & $0.27036285$\\
$5.0$ & $1.0000000000$ & $1.0000000000$ & $0.10332024$ & $0.10332024$\\
$20$  & $1.0000000000$ & $1.0000000000$ & $0.02520828$ & $0.02520828$\\
\bottomrule
\end{tabular}
\end{table}

\subsection*{A.7 Consequence for the main text}

Conditioning this prior on the physical region $S\ge0$ --- unavoidable by
Proposition~\ref{prop:exist} --- and comparing with the exponential prior
\eqref{eq:prior} at equal mean dissipation gives Table~\ref{tab:ftprior}. Two
conventions are needed to read that table and are stated here rather than left
implicit. The row label $\mu/k_B$ is the mean of the prior \emph{before}
conditioning, which fixes the multiplier $a$; the conditioned mean is the larger
number $m_{+}=\langle u\mid u\ge0\rangle$, and it is $m_{+}$, not $\mu$, at which
the exponential prior is matched. And $\alpha=T_0\langle S\rangle/W$ needs a work
output as well as a mean dissipation, so the table fixes $W=2T_0\mu$, giving
$\alpha=m_{+}/2\mu$. With those choices the two priors agree to $3.0\%$ at
$\mu=0.5\,k_B$ and to $0.09\%$ at $\mu=10\,k_B$, the relative difference being
taken against the exponential value, and the agreement improving monotonically as
the mean dissipation grows --- to $0.014\%$ at $\mu=30\,k_B$. The mechanism is
that as $a\to0$ the conditioned density tends to $a\,e^{-au}$, which is the
exponential prior itself.
Admitting negative entropy production in a fluctuation-theorem-consistent way
therefore changes the results of Sections~3 and~4 negligibly. What cannot be
relaxed is the conditioning on $S\ge0$ itself, which Proposition~\ref{prop:exist}
shows is required for the mean to exist at all.

\begin{table}[H]
\centering
\caption{Mean efficiency under the fluctuation-theorem prior conditioned on
$S\ge0$, against the exponential prior \eqref{eq:prior} at the same effective
$\alpha$. The row label $\mu/k_B$ is the \emph{unconditioned} mean, which fixes
the multiplier; the conditioned mean is larger, and the comparison is made at the
conditioned mean with $W=2T_0\mu$, so that $\alpha=m_{+}/2\mu$ (Section~A.7). The
final column is relative to the exponential value. Agreement improves
monotonically as $\mu/k_B$ grows}
\label{tab:ftprior}
\begin{tabular}{ccccc}
\toprule
$\mu/k_B$ & $\alpha$ & FT prior & exponential prior & difference \\
\midrule
$0.5$  & $1.143$ & $0.5531$ & $0.5705$ & $3.0\%$ \\
$1$    & $0.929$ & $0.5994$ & $0.6104$ & $1.8\%$ \\
$2$    & $0.785$ & $0.6365$ & $0.6422$ & $0.9\%$ \\
$5$    & $0.664$ & $0.6713$ & $0.6732$ & $0.3\%$ \\
$10$   & $0.605$ & $0.6890$ & $0.6896$ & $0.09\%$ \\
\bottomrule
\end{tabular}
\end{table}

\section*{Statements and Declarations}

\subsection*{Data and code availability}
This work reports no experimental data. All numerical results, figures and
verification checks reported in the paper are reproducible from the scripts
provided as Supplementary Information, which regenerate every figure and every
numerical value quoted in the text, re-derive each proposition symbolically, and
test the bounds against randomly generated distributions. The same scripts are
archived, in the exact state that produced the numbers reported here, at
\url{https://doi.org/10.5281/zenodo.22812101}.


\end{document}